\documentclass[11pt,amssymb,amsfont,a4paper]{article}
\usepackage{latexsym,amssymb,amsmath,amsthm,color}
\usepackage{geometry}
\usepackage{url}
\usepackage[utf8]{inputenc}
\usepackage{authblk}
\usepackage{bbm}

\theoremstyle{definition}
\newtheorem{definition}{Definition}
\newtheorem{example}{Example}

\theoremstyle{plain}
\newtheorem{theorem}{Theorem}
\newtheorem{proposition}{Proposition}
\newtheorem{lemma}{Lemma}
\newtheorem{remark}{Remark}
\newtheorem{corollary}{Corollary}

\title{On the roots and minimum rank distance of skew cyclic codes}
\author{Umberto Mart{\'i}nez-Pe\~{n}as \thanks{umberto@math.aau.dk}}
\affil{Department of Mathematical Sciences, Aalborg University, Denmark}

\begin{document}

\date{}

\maketitle

\begin{abstract}
Skew cyclic codes play the same role as cyclic codes in the theory of error-correcting codes for the rank metric. In this paper, we give descriptions of these codes by root spaces, cyclotomic spaces and idempotent generators. We prove that the lattice of skew cyclic codes is anti-isomorphic to the lattice of root spaces, study these two lattices and extend the rank-BCH bound on their minimum rank distance to rank-metric versions of the van Lint-Wilson's shift and Hartmann-Tzeng bounds. Finally, we study skew cyclic codes which are linear over the base field, proving that these codes include all Hamming-metric cyclic codes, giving then a new relation between these codes and rank-metric skew cyclic codes.


\end{abstract}

\section{Introduction} \label{introduction}

Cyclic codes play a very important role in the theory of error-correcting codes in the Hamming metric. On the other hand, error-correcting codes in the rank metric \cite{gabidulin} have been proven to be crucial in applications to network coding (see \cite{on-metrics}). Only a few families of rank-metric codes are known (for instance \cite{gabidulin,new-construction}) and only for a restricted choice of parameters. Therefore it is of interest to study new and different families of codes with good rank-metric parameters, simple algebraic descriptions and fast encoding and decoding algorithms.

Usual cyclic codes have been considered for the rank metric in \cite{oggiercyclic,rankdistancecyclic} and a new construction, the so-called rank $ q $-cyclic codes, was introduced in \cite{gabidulin} for square matrices and has been generalized in \cite{qcyclic} for other lengths. Independently, this notion has been generalized to skew or $ q^r $-cyclic codes in the work by Ulmer et al. in \cite{skewcyclic1,skewcyclic2,skewcyclic3}, where $ r $ may be different from $ 1 $.

Some Gabidulin codes consisting of square matrices are $ q $-cyclic (see \cite{gabidulin,qcyclic}), which implies that the family of $ q $-cyclic codes includes some maximum rank distance (MRD) codes. In \cite{gabidulin}, in \cite{qcyclic} and in \cite{skewcyclic1,skewcyclic2,skewcyclic3}, it is also shown (in increasing order of generality) that these codes can be represented as left ideals in a quotient ring of linearized polynomials. Therefore, this construction of rank-metric codes seems to be the appropriate extension of cyclic codes to the rank metric.

In this paper, we focus on two objectives: First, studying the minimum rank distance of skew cyclic codes by giving new lower bounds and by relating it with the Hamming metric. Secondly, studying and relating the lattices of skew cyclic codes and root spaces, which in particular allows to easily construct skew cyclic codes and compare the sharpness of the obtained bounds.

After some preliminaries in Section \ref{section preliminaries}, the results are organized as follows: In Section \ref{section roots}, we give descriptions of skew cyclic codes by root spaces and cyclotomic spaces. In Section \ref{section lattices}, we prove that the lattices of skew cyclic codes and root spaces are anti-isomorphic (isomorphic with the orders reversed), and study these lattices. In Section \ref{section bounds}, we give bounds on their minimum rank distance, extending the rank-BCH bound obtained in \cite{skewcyclic3} to rank-metric versions of the Hartmann-Tzeng bound \cite{hartmann} and the van Lint-Wilson shift bound \cite{vanlint}. Finally, in Section \ref{general cyclic}, we study skew cyclic codes that are linear over the base field, proving that classical cyclic codes equipped with the Hamming metric are a particular case of skew cyclic codes equipped with the rank metric, giving then new relations between both.

\section{Definitions and preliminaries} \label{section preliminaries}

\subsection{Finite field extensions used in this work}

Fix from now on a prime power $ q $ and positive integers $ m $, $ n $ and $ r $, and assume that $ m $ divides $ rn $. We will consider the four finite fields $ \mathbb{F}_q $, $ \mathbb{F}_{q^r} $, $ \mathbb{F}_{q^m} $ and $ \mathbb{F}_{q^{rn}} $ shown in the following graph, where $ \mathbb{F} \longrightarrow \mathbb{F}^\prime $ means that $ \mathbb{F}^\prime $ is an extension of $ \mathbb{F} $:

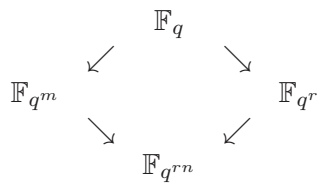
\begin{displaymath}
\begin{array}{rcccl}
 & & \mathbb{F}_q & & \\
  & \swarrow & & \searrow & \\
 \mathbb{F}_{q^m} & & & & \mathbb{F}_{q^r} \\
 & \searrow & & \swarrow & \\
  & & \mathbb{F}_{q^{rn}} & & \\
\end{array}
\end{displaymath}

Dimensions of vector spaces over a field $ \mathbb{F} $ will be denoted by $ \dim_{\mathbb{F}} $, or just $ \dim $ if the field is clear from the context. For a field extension $ \mathbb{F} \subseteq \mathbb{F}^\prime $ and a subset $ A \subseteq \mathbb{F}^{\prime n} $, we denote by $ \langle A \rangle_{\mathbb{F}} $ the $ \mathbb{F} $-linear vector space in $ \mathbb{F}^{\prime n} $ generated by $ A $.

\subsection{Rank-metric codes and generalized Gabidulin codes} \label{subsec rank codes}

For convenience, all coordinates from $ 0 $ to $ n - 1 $ or $ m - 1 $ will be considered as integers modulo $ n $ or $ m $, respectively. Given $ \mathbf{c} = (c_0, c_1, \ldots, c_{n-1}) \in \mathbb{F}_{q^m}^n $, its rank weight \cite{gabidulin} is defined as $ {\rm wt_R}(\mathbf{c}) = \dim_{\mathbb{F}_q}(\langle c_0, c_1, \ldots, c_{n-1} \rangle_{\mathbb{F}_q}) $. Equivalently, if $ \alpha_0, \alpha_1, \ldots, \alpha_{m-1} $ is a basis of $ \mathbb{F}_{q^m} $ over $ \mathbb{F}_q $ and $ \mathbf{c} = \sum_{i=0}^{m-1} \alpha_i \mathbf{c}_i $, where $ \mathbf{c}_i \in \mathbb{F}_q^n $, then $ {\rm wt_R}(\mathbf{c}) = \dim_{\mathbb{F}_q}(\langle \mathbf{c}_0, \mathbf{c}_1, \ldots, $ $ \mathbf{c}_{m-1} \rangle_{\mathbb{F}_q}) $.

For an $ \mathbb{F}_{q^m} $-linear code $ C \subseteq \mathbb{F}_{q^m}^n $, its minimum rank distance is $ d_R(C) = \min \{ {\rm wt_R}(\mathbf{c}) \mid \mathbf{c} \in C \setminus \{ \mathbf{0} \} \} $. We have the Singleton bound \cite{gabidulin} $ d_R(C) \leq n - \dim(C) + 1 $, and we say that $ C $ is maximum rank distance (MRD) if equality holds.

Sometimes we will use a normal basis, that is, a basis of $ \mathbb{F}_{q^m} $ (or $ \mathbb{F}_{q^n} $) over $ \mathbb{F}_q $ of the form $ \alpha, \alpha^{[1]}, \alpha^{[2]}, \ldots, \alpha^{[m-1]} $ (or $ \alpha^{[n-1]} $), for some $ \alpha \in \mathbb{F}_{q^m} $, where we use the notation $ [i] = q^i $. Normal bases exist for all values of $ m $ (or $ n $). See for instance, \cite[Theorem 3.73]{lidl}.

We will consider the following family of MRD codes, usually called Gabidulin codes. They were originally defined in \cite{gabidulin} for $ r=1 $, and generalized for any $ r $ in \cite{new-construction}. Assume that $ n \leq m $ and $ r $ and $ m $ are coprime, and take a vector $ \boldsymbol\beta = (\beta_0,\beta_1, \ldots, \beta_{n-1}) \in \mathbb{F}_{q^m}^{n} $, where $ \beta_0,\beta_1, \ldots, \beta_{n-1} $ are linearly independent over $ \mathbb{F}_q $, and an integer $ 1 \leq k \leq n $. We define the (generalized) Gabidulin code of dimension $ k $ in $ \mathbb{F}_{q^m}^n $ as the $ \mathbb{F}_{q^m} $-linear code $ {\rm Gab}_{k,r}(\boldsymbol\beta) $ with parity check matrix given by
\begin{displaymath}
\mathcal{H}_{k,r}(\boldsymbol\beta) = \left(
\begin{array}{ccccc}
\beta_0 & \beta_1 & \beta_2 & \ldots & \beta_{n-1} \\
\beta_0^{[r]} & \beta_1^{[r]} & \beta_2^{[r]} & \ldots & \beta_{n-1}^{[r]} \\
\beta_0^{[2r]} & \beta_1^{[2r]} & \beta_2^{[2r]} & \ldots & \beta_{n-1}^{[2r]} \\
\vdots & \vdots & \vdots & \ddots & \vdots \\
\beta_{0}^{[(n-k-1)r]} & \beta_{1}^{[(n-k-1)r]} & \beta_{2}^{[(n-k-1)r]} & \ldots & \beta_{n-1}^{[(n-k-1)r]} \\
\end{array} \right).
\end{displaymath}

\subsection{Linearized polynomials}

Denote by $ \mathcal{L}_{q^r}\mathbb{F}_{q^m}[x] $ the set of $ q^r $-linearized polynomials (abbreviated as $ q^r $-polynomials) over $ \mathbb{F}_{q^m} $ (see \cite{gabidulin,orespecial,ore} or \cite[Chapter 3]{lidl}), that is, the polynomials in $ x $ of the form
$$ F(x) = F_0 x + F_1 x^{[r]} + F_2 x^{[2r]} + \cdots + F_d x^{[dr]}, $$
where $ F_0, F_1, \ldots, F_d \in \mathbb{F}_{q^m} $, for $ i = 0,1,2, \ldots, d $. We will denote $ \deg_{q^r}(F(x)) = d $ if $ F_d \neq 0 $ and consider the symbolic product $ \otimes $ in $ \mathcal{L}_{q^r}\mathbb{F}_{q^m}[x] $, defined as follows
$$ F(x) \otimes G(x) = F(G(x)), $$
for any $ F(x),G(x) \in \mathcal{L}_{q^r}\mathbb{F}_{q^m}[x] $ (see \cite{gabidulin,lidl,orespecial,ore}). This product is distributive with respect to usual addition, associative, non-commutative and $ x $ is a left and right unit. Endowed with it and usual addition, $ \mathcal{L}_{q^r}\mathbb{F}_{q^m}[x] $ is a left and right Euclidean domain, that is, left and right Euclidean divisions exist (see \cite{orespecial,ore}). The term ``product'' will mean ``symbolic product'', and we will use the term ``conventional'' for the usual product.

\subsection{Skew cyclic codes: Generator and check polynomials}

\begin{definition}[\textbf{\cite{skewcyclic1,skewcyclic2,skewcyclic3,gabidulin,qcyclic}}] \label{definition qcyclic}
Let $ C \subseteq \mathbb{F}_{q^m}^n $ be an arbitrary (linear or non-linear) code. We say that it is skew cyclic or $ q^r $-cyclic if the $ q^r $-shifted vector
\begin{equation} \label{shifting operator}
\sigma_{r,n}(\mathbf{c}) = (c_{n-1}^{[r]}, c_0^{[r]}, c_1^{[r]}, \ldots, c_{n-2}^{[r]})
\end{equation}
lies in $ C $, for every $ \mathbf{c} = (c_0,c_1, \ldots, c_{n-1}) \in C $.
\end{definition}

Observe that we may assume that $ 1 \leq r \leq m $. Moreover, by taking $ r = m $, we recover the definition of cyclic codes.

Since $ m $ divides $ rn $, $ x^{[rn]} - x $ commutes with every $ q^r $-polynomial in $ \mathcal{L}_{q^r}\mathbb{F}_{q^m}[x] $ and we may consider the quotient ring $ \mathcal{L}_{q^r}\mathbb{F}_{q^m}[x] / (x^{[rn]} - x) $, isomorphic as $ \mathbb{F}_{q^m} $-linear vector space to $ \mathbb{F}_{q^m}^n $ by the map $ \gamma_r : \mathbb{F}_{q^m}^n \longrightarrow \mathcal{L}_{q^r}\mathbb{F}_{q^m}[x] / (x^{[rn]} - x) $, where
\begin{equation} \label{definition gamma}
\gamma_r (F_0, F_1, \ldots, F_{n-1}) = F_0 x + F_1 x^{[r]} + F_2 x^{[2r]} + \cdots + F_{n-1}x^{[(n-1)r]}.
\end{equation}
In the rest of the paper, given $ F(x) \in \mathcal{L}_{q^r}\mathbb{F}_{q^m}[x] $, we will use the notation $ F $ for the class of $ F(x) $ modulo $ x^{[rn]} - x $, that is, for the element $ F = F(x) + (x^{[rn]} - x) \in \mathcal{L}_{q^r}\mathbb{F}_{q^m}[x] / (x^{[rn]} - x) $. 

For $ C \subseteq \mathbb{F}_{q^m}^n $, we define $ C(x) = \gamma_r(C) $, that is, the image of $ C $ by the map $ \gamma_r $. The following characterization is obtained independently in \cite[Theorem 1]{skewcyclic1} and \cite[Lemma 3]{qcyclic}:

\begin{lemma} [\textbf{\cite{skewcyclic1,qcyclic}}] \label{char qcyclic}
A code $ C \subseteq \mathbb{F}_{q^m}^n $ is $ \mathbb{F}_{q^m} $-linear and $ q^r $-cyclic if, and only if, $ C(x) $ is a left ideal in $ \mathcal{L}_{q^r}\mathbb{F}_{q^m}[x] / (x^{[rn]} - x) $.
\end{lemma}

\begin{remark}
In \cite{skewcyclic1,skewcyclic2,skewcyclic3}, and in \cite{qcyclic} for $ r=1 $, left ideals in the rings $ \mathcal{L}_{q^r}\mathbb{F}_{q}[x] / (L(x)) $ are also considered, where $ L(x) $ commutes with every other $ q^r $-polynomial in $ \mathcal{L}_{q^r}\mathbb{F}_{q^m}[x] $. We will call these codes pseudo-$ q^r $-cyclic codes. The results in this paper concerning $ q^r $-root spaces and left ideals in $ \mathcal{L}_{q^r}\mathbb{F}_{q^m}[x] / (x^{[rn]} - x) $ may be directly generalized to left ideals in $ \mathcal{L}_{q^r}\mathbb{F}_{q}[x] / (L(x)) $, if $ L(x) $ has simple roots and if we replace $ \mathbb{F}_{q^{rn}} $ by the splitting field of $ L(x) $. The results are written for $ L(x) = x^{[rn]} - x $ for simplicity.
\end{remark}

From now on, we will fix a left ideal $ C(x) $. The following theorem summarizes the main properties of the generator and check polynomials of $ C $. These were proven in \cite{skewcyclic2}, in \cite{qcyclic} for $ r=1 $, and originally in \cite{gabidulin} for $ r=1 $ and $ m=n $:

\begin{theorem} [\textbf{Generator and check polynomials \cite{skewcyclic2,gabidulin,qcyclic}}] \label{generator}
There exists a unique $ q^r $-polynomial $ G(x) = G_0x + G_1 x^{[r]} + \cdots + G_{n-k} x^{[(n-k)r]} $ over $ \mathbb{F}_{q^m} $ of degree $ q^{(n-k)r} $ that is monic and of minimal degree among the $ q^r $-polynomials whose residue class modulo $ x^{[rn]} - x $ lies in $ C(x) $. It satisfies that $ C(x) = (G) $. There exists another (unique) $ q^r $-polynomial $ H(x) = H_0 x + H_1 x^{[r]} + \cdots + H_k x^{[kr]} $  over $ \mathbb{F}_{q^m} $ such that $ x^{[rn]} - x = G(x) \otimes H(x) = H(x) \otimes G(x) $. They satisfy:
\begin{enumerate}
\item
A $ q^r $-polynomial $ F $ lies in $ C(x) $ if, and only if, $ G(x) $ divides $ F(x) $ on the right (in the ring $ \mathcal{L}_{q^r} \mathbb{F}_{q^m}[x] $).
\item
The $ q^r $-polynomials $ x \otimes G, x^{[r]} \otimes G, \ldots, x^{[(k-1)r]} \otimes G $ constitute a basis (over $ \mathbb{F}_{q^m} $) of $ C(x) $.
\item
The dimension of $ C $ (over $ \mathbb{F}_{q^m} $) is $ k = n - \deg_{q^r}(G(x)) $.
\item
$ C $ has a generator matrix (over $ \mathbb{F}_{q^m} $) given by
\begin{displaymath}
\mathcal{G} = \left(
\begin{array}{ccccccc}
G_0 & G_1 & \ldots & G_{n-k} & 0 & \ldots & 0 \\
0 & G_0^{[r]} & \ldots & G_{n-k-1}^{[r]} & G_{n-k}^{[r]} & \ldots & 0 \\
\vdots & \vdots & \ddots & \vdots & \vdots & \ddots & \vdots \\
0 & 0 & \ldots & G_0^{[(k-1)r]} & G_1^{[(k-1)r]} & \ldots & G_{n-k}^{[(k-1)r]} \\
\end{array} \right).
\end{displaymath}
Moreover, if $ C $ has another generator matrix $ \mathcal{G}^\prime $  with the same form, for the values $ G_i^\prime $, $ i = 0,1,2, \ldots, n-k $, then $ G^\prime_i = G^\prime_{n-k} G_i $, for all $ i $. 
\item
A $ q^r $-polynomial $ F $ lies in $ C(x) $ if, and only if, $ F \otimes H = 0 $.
\item
$ C $ has a parity check matrix (over $ \mathbb{F}_{q^m} $) given by
\begin{displaymath}
\mathcal{H} = \left(
\begin{array}{ccccccc}
h_k & h_{k-1} & \ldots & h_0 & 0 & \ldots & 0 \\
0 & h_k^{[r]} & \ldots & h_{1}^{[r]} & h_{0}^{[r]} & \ldots & 0 \\
\vdots & \vdots & \ddots & \vdots & \vdots & \ddots & \vdots \\
0 & 0 & \ldots & h_k^{[(n-k-1)r]} & h_{k-1}^{[(n-k-1)r]} & \ldots & h_{0}^{[(n-k-1)r]} \\
\end{array} \right),
\end{displaymath}
where $ h_i = H_i^{[(k-i)r]} $. 
\item
$ C^\perp $ is also $ q^r $-cyclic and its generator of minimal degree is $ H^\perp(x) = (h_k x + h_{k-1}x^{[r]} + \cdots + h_0 x^{[kr]}) / h_0 $.
\end{enumerate}
\end{theorem}

The $ q^r $-polynomial $ G(x) $ will be called the minimal generator of $ C(x) $, and $ H(x) $ will be called the minimal check $ q^r $-polynomial of $ C(x) $.

\subsection{The assumptions on the lengths of skew cyclic codes}

To conclude, we see that restricting to the case where $ m $ divides $ rn $ does not leave any $ q^r $-cyclic code out of study. Assume that $ N $ is a positive integer, and take an arbitrary $ q^r $-cyclic code $ C \subseteq \mathbb{F}_{q^m}^N $. Define $ n = {\rm lcm}(m,N) $, which satisfies that $ n = sm = tN $ for positive integers $ s $ and $ t $, and define $ \psi : \mathbb{F}_{q^m}^N \longrightarrow \mathbb{F}_{q^m}^n $ by 
$$ \psi(c_0,c_1, \ldots, c_{N-1}) = (c_0,c_1, \ldots, c_{N-1} ; c_0,c_1, \ldots, c_{N-1} ; \ldots ; c_0,c_1, \ldots, c_{N-1}), $$
where we repeat the vector $ (c_0,c_1, \ldots, c_{N-1}) $ $ t $ times. It holds that $ \psi $ is $ \mathbb{F}_{q^m} $-linear, one to one and $ {\rm wt_R}(\mathbf{c}) = {\rm wt_R}(\psi(\mathbf{c})) $, for all $ \mathbf{c} \in \mathbb{F}_{q^m}^N $. Moreover, if we define $ \sigma_{r,n} $ and $ \sigma_{r,N} $ as in Definition \ref{definition qcyclic}, then $ \psi (\sigma_{r,N}(\mathbf{c})) = \sigma_{r,n}(\psi(\mathbf{c})) $, for all $ \mathbf{c} \in \mathbb{F}_{q^m}^N $, and therefore, $ C \subseteq \mathbb{F}_{q^m}^N $ is $ q^r $-cyclic if, and only if, so is $ \psi(C) $. The same holds for $ \mathbb{F}_q $-linearity and $ \mathbb{F}_{q^m} $-linearity. To sum up, every $ q^r $-cyclic code can be seen as a code in $ \mathbb{F}_{q^m}^n $, where $ m $ divides $ rn $.

\section{Root spaces and cyclotomic spaces} \label{section roots}

In this section we will describe left ideals in $ \mathcal{L}_{q^r}\mathbb{F}_{q^m}[x] / (x^{[rn]}-x) $ in terms of $ q^r $-root spaces and $ q^r $-cyclotomic spaces, which are a subfamily of the former one and will which allow to easily construct skew cyclic codes. As in the classical theory of cyclic codes, we will see that the lattice of $ q^r $-cyclic codes is anti-isomorphic (isomorphic with the orders reversed) to the lattice of $ q^r $-root spaces. In Section \ref{section bounds} we will use this $ q^r $-root space description of $ q^r $-cyclic codes to extend the rank-BCH bound in \cite[Proposition 1]{skewcyclic3} to more general bounds on the minimum rank distance of $ q^r $-cyclic codes.

\subsection{The root space associated to a skew cyclic code}

A $ q^r $-polynomial $ F(x) $ over $ \mathbb{F}_{q^m} $ defines an $ \mathbb{F}_{q^r} $-linear map $ F: \mathbb{F}_{q^{rn}} \longrightarrow \mathbb{F}_{q^{rn}} $, and in particular its set of roots or zeroes in $ \mathbb{F}_{q^{rn}} $ is an $ \mathbb{F}_{q^r} $-linear vector space. 

\begin{definition}[\textbf{Root spaces}] \label{def root space}
An $ \mathbb{F}_{q^r} $-linear subspace of $ \mathbb{F}_{q^{rn}} $ will be called a $ q^r $-root space over $ \mathbb{F}_{q^m} $ if it is the space of roots in $ \mathbb{F}_{q^{rn}} $ of some $ F(x) \in \mathcal{L}_{q^r}\mathbb{F}_{q^m}[x] $.

On the other hand, for a residue class $ F = F(x) + (x^{[rn]}-x) $, we define its root space, denoted as $ Z(F) $, as the root space in $ \mathbb{F}_{q^{rn}} $ of $ F(x) $. 

Finally, define the map $ \rho_r $ between the family of $ \mathbb{F}_{q^m} $-linear $ q^r $-cyclic codes in $ \mathbb{F}_{q^m}^n $ and the family of $ q^r $-root spaces over $ \mathbb{F}_{q^m} $ in $ \mathbb{F}_{q^{rn}} $ by $ \rho_r(C) = T $, where $ T = Z(G) $ and $ G(x) $ is the minimal generator of $ C(x) $.
\end{definition}

Observe that the second definition is consistent, since if $ F_1 = F_2 $, then $ F_1(x) - F_2(x) $ is divisible on the right by $ x^{[rn]} - x $, and hence $ F_1(x) $ and $ F_2(x) $ have the same roots in $ \mathbb{F}_{q^{rn}} $. The following lemma is a particular case of \cite[Theorem 3.50]{lidl}:

\begin{lemma} [\textbf{\cite{lidl}}] \label{root structure}
Given $ F(x) \in \mathcal{L}_{q^r}\mathbb{F}_{q^m}[x] $, assume that the set of its roots $ T $ lie in $ \mathbb{F}_{q^{rn}} $ and all roots are simple. Then
$$ deg_{q^r}(F(x)) = \dim_{\mathbb{F}_{q^r}}(T). $$
\end{lemma}

The following theorem gathers the basic relations between $ C $ and $ \rho_r(C) $:

\begin{theorem} \label{roots}
Let $ T = \rho_r(C) $ as in Definition \ref{def root space}, then:
\begin{enumerate}
\item
$ G(x) = \prod_{\beta \in T} (x - \beta) $.
\item
The dimension of $ C $ over $ \mathbb{F}_{q^m} $ is $ k = n - \dim_{\mathbb{F}_{q^r}}(T) $.
\item
For a $ q^r $-polynomial $ F(x) $, it holds that $ F \in C(x) $ if, and only if, $ F(\beta) = 0 $, for all $ \beta \in T $.
\item
Let $ \beta_1, \beta_2, \ldots, \beta_{n-k} $ be a basis of $ T $ over $ \mathbb{F}_{q^r} $. Then the matrix
\begin{displaymath}
\mathcal{M}(\boldsymbol\beta) = \left(
\begin{array}{ccccc}
\beta_1 & \beta_1^{[r]} & \beta_1^{[2r]} & \ldots & \beta_1^{[(n-1)r]} \\
\beta_2 & \beta_2^{[r]} & \beta_2^{[2r]} & \ldots & \beta_2^{[(n-1)r]} \\
\vdots & \vdots & \vdots & \ddots & \vdots \\
\beta_{n-k} & \beta_{n-k}^{[r]} & \beta_{n-k}^{[2r]} & \ldots & \beta_{n-k}^{[(n-1)r]} \\
\end{array} \right)
\end{displaymath}
is a parity check matrix of $ C $ over $ \mathbb{F}_{q^{rn}} $.
\item
A $ q^r $-polynomial $ \widetilde{G} $ generates $ C(x) $ if, and only if, $ Z(\widetilde{G}) = T $, which holds if, and only if, $ G(x) = {\rm gcd} (\widetilde{G}(x) , x^{[rn]} - x) $ (on the right).
\end{enumerate}
\end{theorem}
\begin{proof}
First, since $ G(x) $ divides $ x^{[rn]}-x $ symbolically on the right, it also divides it conventionally. Therefore, $ G(x) $ has simple roots because $ x^{[rn]} -x $ has simple roots, and item 1 follows.

Since the roots of $ G(x) $ are simple, item 2 follows directly from the previous lemma and Theorem \ref{generator}.

Next, if $ F \in (G) $, then $ G(x) $ divides $ F(x) $ on the right and therefore $ T \subseteq Z(F) $. On the other hand, assume that $ F(\beta) = 0 $, for all $ \beta \in T $. By the Euclidean division, we have that $ F(x) = Q(x) \otimes G(x) + R(x) $, with $ \deg(R(x)) < \deg(G(x)) $, but then $ R(\beta) = 0 $, for all $ \beta \in T $, and hence $ R(x) = 0 $. We conclude that $ F \in (G) $ and item 3 follows. Item 4 follows immediately from item 3.

Finally, assume that $ \widetilde{G} $ generates $ C(x) $. Since $ G $ divides $ \widetilde{G} $ and $ \widetilde{G} $ divides $ G $ on the right, we have that $ Z(\widetilde{G}) = T $. Now assume that $ Z(\widetilde{G}) = T $ and define $ D(x) = {\rm gcd}(\widetilde{G}(x), x^{[rn]} - x) $. We have that $ D(x) = A(x) \otimes \widetilde{G}(x) + B(x) \otimes (x^{[rn]} - x) $, for some $ q^r $-polynomials $ A(x) $ and $ B(x) $. It follows that $ T \subseteq Z(D) $, and since $ D(x) $ divides $ \widetilde{G}(x) $, it holds that $ T = Z(D) $. Finally, since $ D(x) $ divides $ x^{[rn]} - x $, every root of $ D(x) $ lies in $ \mathbb{F}_{q^{rn}} $ and is simple, which implies that $ D(x) = G(x) $. Now assume that $ G(x) = {\rm gcd}(\widetilde{G}(x), x^{[rn]} - x) $, then $ G(x) = A(x) \otimes \widetilde{G}(x) + B(x) \otimes (x^{[rn]} - x) $, for some $ q^r $-polynomials $ A(x) $ and $ B(x) $. Therefore, $ G \in (\widetilde{G}) $, and since $ G(x) $ divides $ \widetilde{G}(x) $, it holds that $ (G) = (\widetilde{G}) $, and item 5 follows.
\end{proof}

On the other hand, we have the following equivalent conditions on inclusions of $ q^r $-cyclic codes and $ q^r $-root spaces.

\begin{corollary} \label{corollary lattices}
Let $ C_1(x) = (G_1) $ and $ C_2(x) = (G_2) $ be two $ q^r $-cyclic codes with $ T_1 = Z(G_1) $ and $ T_2 = Z(G_2) $, where $ G_1(x) $ and $ G_2(x) $ are the minimal generators of $ C_1(x) $ and $ C_2(x) $, respectively. Then $ C_1(x) \subseteq C_2(x) $ if, and only if, $ G_2(x) $ divides $ G_1(x) $ on the right, and this holds if, and only if, $ T_2 \subseteq T_1 $.
\end{corollary}
\begin{proof}
The first equivalence is clear from Theorem \ref{generator}. Now, if $ G_2(x) $ divides $ G_1(x) $ on the right, then it is obvious that $ T_2 \subseteq T_1 $. 

Finally, assume that $ T_2 \subseteq T_1 $, and perform the Euclidean division to obtain $ G_1(x) = Q(x) \otimes G_2(x) + R(x) $, with $ \deg(R(x)) < \deg(G_2(x)) $. We have that $ R(\beta) = 0 $, for every $ \beta \in T_2 $, and by the previous theorem, $ R \in (G_2) $. However, $ G_2(x) $ is the minimal generator of $ C_2(x) $, so it follows that $ R(x)= 0 $, that is, $ G_2(x) $ divides $ G_1(x) $ on the right. 
\end{proof}

The previous corollary and Theorem \ref{roots} imply that the map $ \rho_r $ is bijective:

\begin{corollary} \label{map bijective}
The map $ \rho_r $ in Definition \ref{def root space} is bijective.
\end{corollary}
\begin{proof}
We first see that it is onto. Take $ T = Z(F) $ a $ q^r $-root space over $ \mathbb{F}_{q^m} $ in $ \mathbb{F}_{q^{rn}} $. By item 5 in Theorem \ref{roots}, it holds that $ Z(G) = T $ if $ G(x) $ is the minimal generator of $ C(x) = (F) $. Therefore, $ T = \rho_r(C) $. On the other hand, $ \rho_r $ is one to one by the previous corollary.
\end{proof}

In the next section we will see that the family of $ q^r $-root spaces over $ \mathbb{F}_{q^m} $ in $ \mathbb{F}_{q^{rn}} $ is a lattice with sums and additions of vector spaces, and therefore Corollary \ref{corollary lattices} together with the previous corollary mean that the map $ \rho_r $ is an anti-isomorphism of lattices (an isomorphism with the orders reversed).

On the other hand, Theorem \ref{roots} gives the following criterion to say whether an $ \mathbb{F}_{q^r} $-linear subspace $ T \subseteq \mathbb{F}_{q^{rn}} $ is a $ q^r $-root space, in terms of $ q^r $-cyclic codes:

\begin{corollary} \label{char root spaces}
Let $ T \subseteq \mathbb{F}_{q^{rn}} $ be $ \mathbb{F}_{q^r} $-linear, take one of its bases $ \beta_1, \beta_2, \ldots, \beta_{n-k} $ over $ \mathbb{F}_{q^r} $, and define $ \mathcal{M}(\boldsymbol\beta) $ as in Theorem \ref{roots}. Consider $ \widetilde{C} \subseteq \mathbb{F}_{q^{rn}}^n $, the $ \mathbb{F}_{q^{rn}} $-linear code with $ \mathcal{M}(\boldsymbol\beta) $ as parity check matrix. Then $ T $ is a $ q^r $-root space over $ \mathbb{F}_{q^m} $ if, and only if, 
\begin{equation} \label{condition char root}
\dim_{\mathbb{F}_{q^m}}(\widetilde{C} \cap \mathbb{F}_{q^m}^n) = \dim_{\mathbb{F}_{q^{rn}}}(\widetilde{C}),
\end{equation}
which holds if, and only if, $ \widetilde{C} $ has a basis of vectors in $ \mathbb{F}_{q^m}^n $.
\end{corollary}
\begin{proof}
Assume first that $ T = Z(F) $, for some $ q^r $-polynomial $ F(x) $ over $ \mathbb{F}_{q^m} $, and define $ C(x) = (F) $. By items 4 and 5 in Theorem \ref{roots}, $ C = \widetilde{C} \cap \mathbb{F}_{q^m}^n $, and by item 2 in the same theorem, $ \dim_{\mathbb{F}_{q^m}}(C) = k = \dim_{\mathbb{F}_{q^{rn}}}(\widetilde{C}) $.

Assume now that $ \dim_{\mathbb{F}_{q^m}}(C) = \dim_{\mathbb{F}_{q^{rn}}}(\widetilde{C}) $, where $ C = \widetilde{C} \cap \mathbb{F}_{q^m}^n $. Since $ \widetilde{C} $ is $ q^r $-cyclic, it follows that $ C $ is also $ q^r $-cyclic. By definition, $ T \subseteq Z(G) $, for the minimal generator $ G(x) $ of $ C(x) $. Now, $ \dim_{\mathbb{F}_{q^m}}(C) = k $ by hypothesis, and hence $ \dim_{\mathbb{F}_{q^r}}(Z(G)) = n-k $ by item 2 in Theorem \ref{roots}. Also by hypothesis, $ \dim_{\mathbb{F}_{q^r}}(T) = n-k $, so it holds that $ T = Z(G) $. 
\end{proof}

Observe that condition (\ref{condition char root}) means that $ \widetilde{C} $ is Galois closed over $ \mathbb{F}_{q^m} $. See \cite{similarities,stichtenoth} for more details on Galois closed vector spaces. The following example shows how to use this result to see whether a given vector space is a $ q^r $-root space. 

\begin{example} \label{example qcyclotomic}
Assume that $ n = 2m $ and $ r=1 $, and take a normal basis $ \alpha $, $ \alpha^{[1]}, \ldots $, $ \alpha^{[n-1]} $ $ \in \mathbb{F}_{q^n} $ over $ \mathbb{F}_q $. Consider the ($ \mathbb{F}_q $-linear) vector subspaces $ T_1 , T_2 \subseteq \mathbb{F}_{q^n} $ generated by $ \alpha $ and $ \alpha, \alpha^{[m]} $, respectively. Define also the codes $ \widetilde{C}_1, \widetilde{C}_2 \subseteq \mathbb{F}_{q^n}^n $ with parity check matrices $ \mathcal{M}(\alpha) $ and $ \mathcal{M}(\alpha, \alpha^{[m]}) $, respectively, and define $ D_i = (\widetilde{C}_i \cap \mathbb{F}_{q^m}^n)^\perp $, $ i = 1,2 $. They satisfy $ D_i = {\rm Tr}(\widetilde{C}_i^\perp) $, $ i =1,2 $, by Delsarte's theorem \cite[Theorem 2]{delsarte}, where $ {\rm Tr} $ denotes the trace of the extension $ \mathbb{F}_{q^m} \subseteq \mathbb{F}_{q^n} $, that is, $ {\rm Tr}(x) = x + x^{[m]} $.
 
We will see that $ T_1 $ is not a $ q $-root space over $ \mathbb{F}_{q^m} $, whereas $ T_2 $ is. Moreover, we will see that $ D_1=D_2 $, which has dimension $ 2 $ over $ \mathbb{F}_{q^m} $ and which shows that condition (\ref{condition char root}) in the previous corollary is satisfied for $ T_2 $ but not for $ T_1 $. 

Since $ \dim(T_1) = 1 $, if it were a $ q $-root space, then there would exist $ b \in \mathbb{F}_{q^m} $ with $ F(\alpha)= 0 $, where $ F(x) = x^{[1]} - bx $ by Corollary \ref{map bijective}. Since $ x^{[m]} \otimes F(x) = F(x) \otimes x^{[m]} $, it holds that $ F(\alpha^{[m]}) = 0 $. This would imply that $ \alpha, \alpha^{[m]} \in T_1 $ and $ \dim(T_1) = 1 $, which is absurd.

On the other hand, we see that $ D_1 \subseteq D_2 $. Define the vectors $ \boldsymbol\alpha = (\alpha, \alpha^{[1]}, \ldots, $ $ \alpha^{[n-1]}) \in \mathbb{F}_{q^n}^n $, $ \mathbf{v}_0 = {\rm Tr}(\alpha \boldsymbol\alpha) = \alpha \boldsymbol\alpha + \alpha^{[m]} \boldsymbol\alpha^{[m]} $ and $ \mathbf{v}_1 = {\rm Tr}(\alpha^{[1]} \boldsymbol\alpha) = \alpha^{[1]} \boldsymbol\alpha + \alpha^{[1+m]} \boldsymbol\alpha^{[m]} $, which belong to $ D_1 $ and also to $ \widetilde{C}_2^\perp $. Moreover, we see that they are linearly independent over $ \mathbb{F}_{q^n} $ and, therefore, they constitute a basis of $ \widetilde{C}_2^\perp $. This means that $ D_1=D_2 $ and $ \dim_{\mathbb{F}_{q^m}}(D_2) = \dim_{\mathbb{F}_{q^n}}(\widetilde{C}_2^\perp) = 2 $.

In conclusion, condition (\ref{condition char root}) is satisfied for $ T_2 $ but not for $ T_1 $. By the previous corollary, it holds that $ T_2 $ is a $ q $-root space over $ \mathbb{F}_{q^m} $, and we have seen that $ T_1 $ is not a $ q $-root space over $ \mathbb{F}_{q^m} $.
\end{example}

\subsection{Cyclotomic spaces}

Now we turn to a special subclass of $ q^r $-root spaces in $ \mathbb{F}_{q^{rn}} $, namely the class of $ q^r $-cyclotomic spaces. These spaces will play the same role as cyclotomic sets in the classical theory of cyclic codes (see \cite[Theorem 4.4.2]{pless} and \cite[Theorem 4.4.3]{pless}), that is, they generate the lattice of $ q^r $-root spaces, and are key concepts to easily construct skew cyclic codes.

For this we need the concept of minimal $ q^r $-polynomial of an element $ \beta \in \mathbb{F}_{q^{rn}} $ over $ \mathbb{F}_{q^m} $. The following lemma and definition constitute an extension of \cite[Theorem 3.68]{lidl} and the discussion prior to it:

\begin{lemma} \label{minimal polynomial}
For any $ \beta $ in an extension field of $ \mathbb{F}_{q^{{\rm lcm}(r,m)}} $, there exists a unique monic $ q^r $-polynomial $ F(x) \in \mathcal{L}_{q^r} \mathbb{F}_{q^{m}}[x] $ of minimal degree such that $ F(\beta)= 0 $. Moreover, if $ L(\beta) = 0 $ for another $ q^r $-polynomial $ L(x) $ over $ \mathbb{F}_{q^{m}} $, then $ F(x) $ divides $ L(x) $ both conventionally and symbolically on the right.
\end{lemma}
\begin{proof}
If $ \beta \in \mathbb{F}_{q^{rt}} $, $ t > 0 $, then the polynomial $ \widetilde{F}(x) = x^{[rt]} - x $ lies in $ \mathcal{L}_{q^r} \mathbb{F}_{q^m}[x] $ and $ \widetilde{F}(\beta) = 0 $. Therefore there exists an $ F(x) \in \mathcal{L}_{q^r} \mathbb{F}_{q^m}[x] $ monic and of minimal degree such that $ F(\beta) = 0 $. Let $ L(x) \in \mathcal{L}_{q^r} \mathbb{F}_{q^m}[x] $ be such that $ L(\beta) = 0 $, and perform the Euclidean division to obtain $ L(x) = Q(x) \otimes F(x) + R(x) $, with $ \deg(R(x)) < \deg(F(x)) $. Then $ R(\beta) = 0 $, and since $ F(x) $ is of minimal degree, we have that $ R(x) = 0 $, and therefore $ F(x) $ divides $ L(x) $ both conventionally and symbolically on the right. This also proves that $ F(x) $ is unique and we are done.
\end{proof}

\begin{definition}
For $ \beta $ in an extension field of $ \mathbb{F}_{q^{{\rm lcm}(r,m)}} $, the $ q^r $-polynomial $ F(x) $ in the previous lemma is called the minimal $ q^r $-polynomial of $ \beta $ over $ \mathbb{F}_{q^{m}} $. 
\end{definition}

Now we may define $ q^r $-cyclotomic spaces in $ \mathbb{F}_{q^{rn}} $:

\begin{definition} [\textbf{Cyclotomic spaces}]
Given $ \beta \in \mathbb{F}_{q^{rn}} $, we define its $ q^r $-cyclotomic space over $ \mathbb{F}_{q^m} $ as the $ \mathbb{F}_{q^r} $-linear vector space $ C_{q^r}(\beta) $ of roots of the minimal $ q^r $-polynomial of $ \beta $ over $ \mathbb{F}_{q^{m}} $.
\end{definition}

\begin{example}
Let the notation and assumptions be as in Example \ref{example qcyclotomic}. Since the basis $ \alpha^{[b]}, (\alpha^{[b]})^{[1]}, \ldots, (\alpha^{[b]})^{[n-1]} $ is also normal, in Example \ref{example qcyclotomic} we have proven that $ C_{q}(\alpha^{[b]}) = \langle \alpha^{[b]}, \alpha^{[b+m]} \rangle $.
\end{example}

In general, for $ r=1 $ and $ n=sm $, we have the following result:

\begin{proposition} \label{cyclotomic normal basis}
If $ \alpha, \alpha^{[1]}, \ldots, \alpha^{[n-1]} $ is a normal basis of $ \mathbb{F}_{q^n} $ over $ \mathbb{F}_q $, then it holds that $ C_{q}(\alpha^{[b]}) = \langle \alpha^{[b]}, \alpha^{[b+m]}, \ldots, \alpha^{[b+(s-1)m]} \rangle $, for every integer $ b \geq 0 $.
\end{proposition}
\begin{proof}
We may assume that $ b= 0 $ without loss of generality. First of all, for every $ F(x) \in \mathcal{L}_{q^r} \mathbb{F}_{q^m}[x] $, we see that $ x^{[m]} \otimes F(x) = F(x) \otimes x^{[m]} $ and, therefore $ F(\beta)= 0 $ implies that $ F(\beta^{[m]}) = 0 $, for any $ \beta \in \mathbb{F}_{q^n} $. This means that $ \langle \alpha, \alpha^{[m]}, \ldots, \alpha^{[(s-1)m]} \rangle \subseteq C_{q}(\alpha) $.

The reversed inclusion is proven using Corollary \ref{char root spaces} as in Example \ref{example qcyclotomic}. To that end, we need to define the vectors $ \mathbf{v}_i = {\rm Tr}(\alpha^{[i]} \boldsymbol\alpha) = \sum_{j = 0}^{s-1} \alpha^{[i+jm]} \boldsymbol\alpha^{[jm]} \in \mathbb{F}_{q^m}^n $, for $ i = 0,1,2, \ldots, s-1 $, where $ \boldsymbol\alpha = (\alpha, \alpha^{[1]}, \ldots, \alpha^{[n-1]}) \in \mathbb{F}_{q^n}^n $. The vectors $ \mathbf{v}_0, \mathbf{v}_1, \ldots, \mathbf{v}_{s-1} $ are linearly independent over $ \mathbb{F}_{q^m} $, since so are the vectors $ \boldsymbol\alpha $, $ \boldsymbol\alpha^{[m]}, \ldots, $ $ \boldsymbol\alpha^{[(s-1)m]} $ and the following matrix is non-singular:
\begin{displaymath}
\left( \begin{array}{ccccc}
\alpha & \alpha^{[m]} & \alpha^{[2m]} & \ldots & \alpha^{[(s-1)m]} \\
\alpha^{[1]} & \alpha^{[1+m]} & \alpha^{[1+2m]} & \ldots & \alpha^{[1+(s-1)m]} \\
\vdots & \vdots & \vdots & \ddots & \vdots \\
\alpha^{[s-1]} & \alpha^{[s-1+m]} & \alpha^{[s-1+2m]} & \ldots & \alpha^{[s-1+(s-1)m]} \\
\end{array} \right).
\end{displaymath}
\end{proof}

Next we see that every $ q^r $-root space is a sum of $ q^r $-cyclotomic spaces. Since in the next section we will see that sums and intersections of $ q^r $-root spaces are again $ q^r $-root spaces, this means that the subclass of $ q^r $-cyclotomic spaces generates the lattice of $ q^r $-root spaces:

\begin{proposition} \label{sums cyclotomic}
Given a $ q^r $-root space $ T \subseteq \mathbb{F}_{q^{rn}} $ over $ \mathbb{F}_{q^m} $, there exist $ \beta_1, \beta_2, \ldots, \beta_u $ $ \in T $ such that $ T = C_{q^r}(\beta_1) + C_{q^r}(\beta_2) + \cdots + C_{q^r}(\beta_u) $. Moreover, if the $ q^r $-cyclotomic spaces $ C_{q^r}(\beta_i) $ over $ \mathbb{F}_{q^m} $ are minimal and $ T $ is not a sum of a strict subset of them, then the sum is direct.
\end{proposition}
\begin{proof}
Take $ L(x) \in \mathcal{L}_{q^r} \mathbb{F}_{q^m}[x] $ such that $ T = Z(L) $. For every $ \beta \in T $, if $ F(x) $ is its minimal $ q^r $-polynomial over $ \mathbb{F}_{q^{m}} $, then by Lemma \ref{minimal polynomial}, $ F(x) $ divides $ L(x) $ and, therefore, $ C_{q^r}(\beta) = Z(F) \subseteq Z(L) = T $. This means that $ T = \sum_{\beta \in T} C_{q^r}(\beta) $. Since the sum is finite, the result follows.

Finally, assume that the $ C_{q^r}(\beta_i) $ are minimal and $ T $ is not a sum of a strict subset of them. If there exists $ \beta \in C_{q^r}(\beta_i) \cap (\sum_{j \neq i} C_{q^r}(\beta_j)) $ that is not zero, then by minimality of $ C_{q^r}(\beta_i) $, we have that $ C_{q^r} (\beta) = C_{q^r}(\beta_i) $, and therefore $ C_{q^r}(\beta_i) \subseteq \sum_{j \neq i} C_{q^r}(\beta_j) $. However, this means that $ T $ is the sum of the spaces $ C_{q^r}(\beta_j) $, with $ j \neq i $, which contradicts the assumptions.
\end{proof}

\section{The lattices of $ q^r $-cyclic codes and $ q^r $-root spaces} \label{section lattices}

It is straightforward to see that sums and intersections of $ q^r $-cyclic codes are again $ q^r $-cyclic. In this section we will see that the same holds for $ q^r $-root spaces. By Corollary \ref{corollary lattices}, both lattices are anti-isomorphic. We will also prove this directly by showing that intersections of $ q^r $-cyclic codes correspond to sums of $ q^r $-root spaces and viceversa. We will also study the concept of $ q^r $-cyclic complementary of a $ q^r $-cyclic code, rank equivalences and lattice morphisms.

\subsection{The lattice anti-isomorphism}

\begin{theorem} \label{lattice}
Let $ C_1(x) $ and $ C_2(x) $ be two $ q^r $-cyclic codes with minimal generators $ G_1(x) $ and $ G_2(x) $, respectively. Set $ T_1 = Z(G_1) $ and $ T_2 = Z(G_2) $. We have that
\begin{enumerate}
\item
$ C_1(x) \cap C_2(x) $ is the $ q^r $-cyclic code whose minimal generator is given by $ M(x) = {\rm lcm}(G_1(x), G_2(x)) $ (on the right), and $ Z(M) = T_1 + T_2 $.
\item
$ C_1(x) + C_2(x) $ is the $ q^r $-cyclic code whose minimal generator is given by $ D(x) = {\rm gcd}(G_1(x), G_2(x)) $ (on the right), and $ Z(D) = T_1 \cap T_2 $.
\end{enumerate}
In particular, sums and intersections of $ q^r $-root spaces are again $ q^r $-root spaces, and they form a lattice anti-isomorphic to the lattice of $ q^r $-cyclic codes by the map $ \rho_r $ in Definition \ref{def root space}. Moreover, the lattice of $ q^r $-root spaces is generated by the subclass of $ q^r $-cyclotomic spaces.
\end{theorem}
\begin{proof}
Define $ M(x) $ as the minimal generator of $ C_1(x) \cap C_2(x) $. We have that $ G_1(x) $ and $ G_2(x) $ both divide $ M(x) $ on the right by Theorem \ref{generator}, item 1, since $ M \in (G_1) $ and $ M \in (G_2) $. Now, if $ F \in C_1(x) \cap C_2(x) $, then $ M(x) $ divides $ F(x) $ on the right for the same reason. In conclusion, $ M(x) $ is the least common multiple on the right of $ G_1(x) $ and $ G_2(x) $.

On the other hand, define $ D(x) $ as the greatest common divisor of $ G_1(x) $ and $ G_2(x) $ on the right. By the Euclidean algorithm, we may find a B{\'e}zout's identity on the right $ D(x) = Q_1(x) \otimes G_1(x) + Q_2(x) \otimes G_2(x) $. This implies that $ (D) \subseteq C_1(x) + C_2(x) $. Moreover, by definition $ D(x) $ divides both $ G_1(x) $ and $ G_2(x) $ on the right, and therefore $ C_1(x) + C_2(x) \subseteq (D) $, and hence they are equal.

To see that $ D(x) $ is the minimal generator, take $ F \in (D) $, then $ F(x) = Q(x) \otimes D(x) + P(x) \otimes (x^{[rn]}-x) $. But since $ D(x) $ divides both $ G_1(x) $ and $ G_2(x) $, and these divide $ x^{[rn]} - x $, then $ D(x) $ divides $ x^{[rn]} - x $ and hence, it divides $ F(x) $.

Finally, we see that $ T_1 \cup T_2 \subseteq Z(M) $ by Theorem \ref{roots}, item 3, since $ M \in C_1(x) \cap C_2(x) $. Therefore, $ T_1 + T_2 \subseteq Z(M) $. On the other hand, since $ D \in C_1(x) + C_2(x) $, we see that $ T_1 \cap T_2 \subseteq Z(D) $ also by Theorem \ref{roots}, item 3. By the same theorem, we have that
$$ \dim(T_1 + T_2) + \dim(T_1 \cap T_2) = \dim(T_1) + \dim(T_2) = (n-\dim(C_1)) + (n-\dim(C_2)) $$
$$ = (n- \dim(C_1 \cap C_2)) + (n - \dim(C_1 + C_2)) = \dim(Z(M)) + \dim(Z(D)). $$
Hence, $ Z(M) = T_1 + T_2 $ and $ Z(D) = T_1 \cap T_2 $ and we are done.

The last statement of the theorem follows from Proposition \ref{sums cyclotomic}. 
\end{proof}

\subsection{Skew cyclic complementaries and idempotent generators}

The existence and/or uniqueness of complementaries is an important property of lattices. In the theory of classical cyclic codes, every cyclic code has a unique complementary cyclic code when the length and $ q $ are coprime \cite[Exercise 243]{pless}. In this case, every cyclic code also has an idempotent generator \cite[Theorem 4.3.2]{pless}, which describes very easily the complementary cyclic code (see \cite[Theorem 4.4.6]{pless}). 

In this subsection we investigate the existence and uniqueness of $ q^r $-cyclic complementaries and idempotent generators of $ q^r $-cyclic codes, and relate both. 

Observe that, by the fact that the map $ \rho_r $ in Definition \ref{roots} is a lattice anti-isomorphism, two $ q^r $-cyclic codes are complementary if, and only if, their corresponding $ q^r $-root spaces are complementary.

\begin{proposition} \label{pre complementaries}
Given $ q^r $-cyclic codes $ C_1(x) $ and $ C_2(x) $ with minimal generators $ G_1(x) $ and $ G_2(x) $, we have that they are complementary, that is, $ \mathbb{F}_{q^m}^n = C_1 \oplus C_2 $ if, and only if, $ G_1(x) $ and $ G_2(x) $ are coprime (on the right) and $ \deg_{q^r}(G_1(x)) + \deg_{q^r}(G_2(x)) = n $. 
\end{proposition}
\begin{proof}
By Theorem \ref{lattice}, the condition $ C_1(x) + C_2(x) = \mathcal{L}_{q^r} \mathbb{F}_{q^m}[x] / (x^{[rn]} - x) $ is equivalent to $ D(x) = x $, which means that $ G_1(x) $ and $ G_2(x) $ are coprime. By Theorem \ref{generator}, if $ C_1 $ and $ C_2 $ are complementary, then 
$$ \deg_{q^r}(G_1(x)) + \deg_{q^r}(G_2(x)) = n - \dim(C_1) + n - \dim(C_2) $$
$$ = n - (\dim(C_1) + \dim(C_2) - \dim(C_1 + C_2)) = n - \dim(C_1 \cap C_2) = n. $$
Conversely, if $ D(x) = x $ and $ \deg_{q^r}(G_1(x)) + \deg_{q^r}(G_2(x)) = n $, then $ C_1 + C_2 = \mathbb{F}_{q^m}^n $ by Theorem \ref{lattice} and $ \dim(C_1 \cap C_2) = 0 $ by Theorem \ref{generator} as before, and the theorem follows.
\end{proof}

In \cite[Theorem 6]{gursoy}, the existence of an idempotent generator is proven when $ n $ is coprime with $ q $ and also with the order of the automorphism $ \alpha \mapsto \alpha^{[r]} $. We next prove the existence in other cases (see Example \ref{example idempotent} below), and give other properties.

\begin{theorem} \label{general idempotent}
Let $ C(x) $ be a left ideal with minimal generator $ G(x) $ and check $ q^r $-polynomial $ H(x) $. The following holds
\begin{enumerate}
\item
An element $ E \in C(x) $ is idempotent (that is, $ E \otimes E = E $) and generates $ C(x) $ if, and only if, it is a unit on the right in this ideal.
\item
Given a $ q^r $-polynomial $ F(x) $ and an idempotent generator $ E $ of $ C(x) $, it holds that $ F \in C(x) $ if, and only if, $ F = F \otimes E $. In particular, $ x - E(x) $ is a check polynomial for $ C(x) $.
\item
For any idempotent generator $ E $ of $ C(x) $, the $ q^r $-polynomial $ x-E $ is also idempotent and $ (x-E) $ is a complementary for $ C(x) $.
\item
Assume that $ G $ and $ H $ are coprime on both sides. That is, we may obtain B{\'e}zout identities on both sides
$$ x = G \otimes G_1 + H \otimes H_1 = G_2 \otimes G + H_2 \otimes H, $$
in the ring $ \mathcal{L}_{q^r} \mathbb{F}_{q^m}[x] / (x^{[rn]}-x) $. Let $ E = x - H_2 \otimes H $ and $ E^\prime = x - H \otimes H_1 $. It holds that $ E=E^\prime $, and it is an idempotent generator for $ C(x) $.
\end{enumerate}
\end{theorem}
\begin{proof}
Items 1 and 2 are proven as in the classical case (see \cite[Section 4.3]{pless}). For item 3, we have that $ (x - E) + (E) $ is the whole quotient ring. On the other hand, take $ F \in (x - E) \cap (E) $. By item 1, $ E $ and $ x - E $ are units on the right in the ideals that they generate. Therefore, $ F = F \otimes E $ and $ F = F \otimes (x - E) = F - F \otimes E = F - F = 0 $. It follows that $ (x - E) \cap (E) = \{ 0 \} $, and item 3 is proven. 

We now prove item 4. We have that $ E = G_2 \otimes G $, $ E^\prime = G \otimes G_1 $ and $ G \otimes H = H \otimes G = 0 $ by Theorem \ref{generator}. Therefore $ E^\prime = E \otimes E^\prime = E $, and it is idempotent. On the other hand, $ E \in (G) $ and $ G = G \otimes E^\prime \in (E^\prime) $, and therefore $ C(x) = (G) = (E) $. 
\end{proof}

From the previous theorem and proposition, we deduce the following for a left ideal $ C(x) $ with minimal generator $ G(x) $ and check $ q^r $-polynomial $ H(x) $:

\begin{corollary} \label{complementaries}
The $ q^r $-cyclic codes $ (G) $ and $ (H) $ are complementary if, and only if, $ G(x) $ and $ H(x) $ are coprime. In that case, if $ E $ is the idempotent described in item 4 in the previous theorem, then $ (x- E) = (H) $.
\end{corollary}

\begin{remark}
Recall from Theorem \ref{roots}, item 5, that in particular, the minimal generator of a left ideal can be efficiently obtained from the idempotent generator.
\end{remark}

\begin{example} \label{example idempotent}
Let $ q=2 $, $ n=m=3 $ and $ r=1 $, consider the primitive element $ \alpha \in \mathbb{F}_{2^3} $ such that $ \alpha^3 + \alpha + 1 = 0 $, and the $ q $-polynomials $ G(x) = x^{[2]} + \alpha^4 x^{[1]} + \alpha^6 x $ and $ H(x) = x^{[1]} + \alpha x $, as in \cite[Example 2]{qcyclic}. By Euclidean division on both sides, we find that
$$ x = x \otimes G(x) + (x^{[1]} + \alpha x) \otimes H(x) = G(x) \otimes x + H(x) \otimes  (x^{[1]} + \alpha x). $$
Then $ E = E^\prime = G $. In this case the idempotent generator coincides with the minimal generator. Observe also that here the order of the automorphism $ \alpha \mapsto \alpha^{[1]} $ is $ 3 $, and hence is not coprime with $ n $. Therefore, Theorem \ref{general idempotent} covers other cases than \cite[Theorem 6]{gursoy}.

On the other hand, we see that the $ q $-polynomial $ x - E = x^{[2]} + \alpha^4 x^{[1]} + \alpha^2 x = (x^{[1]} + \alpha x) \otimes H(x) $ is an idempotent generator of $ (H) $, which is a complementary for $ C(x) $, as stated in the previous corollary.
\end{example}

\subsection{Rank equivalences and lattice automorphisms}

To conclude the section, we study rank equivalences and automorphisms of lattices of the family of $ q^r $-cyclic codes. A rank equivalence $ \varphi : \mathbb{F}_{q^m}^n \longrightarrow \mathbb{F}_{q^m}^n $ is an $ \mathbb{F}_{q^m} $-linear vector space isomorphism with $ {\rm wt_R}(\varphi(\mathbf{c})) = {\rm wt_R}(\mathbf{c}) $ (see \cite{similarities} for more details on rank equivalences). For convenience, we define the rank weight of $ F \in \mathcal{L}_{q^r} \mathbb{F}_{q^m}[x] / (x^{[rn]}-x) $ as 
\begin{equation} \label{definition weight polynomial}
{\rm wt_R}(F) = {\rm wt_R}(F_0,F_1, \ldots, F_{n-1}) = {\rm wt_R}(\gamma_r^{-1}(F)),
\end{equation}
where $ \gamma_r $ is as in (\ref{definition gamma}). Since the map $ \rho_r $ in Definition \ref{roots} is a lattice anti-isomorphism by Theorem \ref{lattice}, every automorphism of the lattice of $ \mathbb{F}_{q^m} $-linear $ q^r $-cyclic codes induces an automorphism of the lattice of $ q^r $-root spaces over $ \mathbb{F}_{q^m} $. In particular, every ring automorphism of $ \mathcal{L}_{q^r} \mathbb{F}_{q^m}[x]/(x^{[rn]} - x) $ induces such a lattice automorphism.

We study the following class of ring automorphisms:

\begin{definition}
For every $ a = 0,1,2, \ldots, rn-1 $, we define the morphism $ \varphi_a : \mathcal{L}_{q^r} \mathbb{F}_{q^m}[x] / $ $ (x^{[rn]}-x) \longrightarrow \mathcal{L}_{q^r} \mathbb{F}_{q^m}[x] / (x^{[rn]}-x) $ by $ \varphi_a (F) = x^{[rn-a]} \otimes F \otimes x^{[a]} $.
\end{definition}

We observe that this map is well-defined and corresponds to rising to the power $ q^{rn-a} $ in $ \mathbb{F}_{q^m}^n $ (and $ \varphi_0 $ is the identity). That is, if $ F = F_0 x + F_1 x^{[r]} + \cdots + F_{n-1} x^{[(n-1)r]} $, then
$$ x^{[rn-a]} \otimes F \otimes x^{[a]} = F_0^{[rn-a]} x + F_1^{[rn-a]} x^{[r]} + \cdots + F_{n-1}^{[rn-a]} x^{[(n-1)r]}. $$
We gather the main properties of the maps $ \varphi_a $ in the next proposition:

\begin{proposition} \label{maps a}
For every $ a, a^\prime = 0,1,2, \ldots, rn-1 $, the map $ \varphi_a $ satisfies:
\begin{enumerate}
\item
$ \varphi_a $ is a ring isomorphism. Viewed as map $ \varphi_a : \mathbb{F}_{q^m}^n \longrightarrow \mathbb{F}_{q^m}^n $, it is $ \mathbb{F}_q $-linear and $ \mathbb{F}_{q^m} $-semilinear.
\item
$ \varphi_a = \varphi_{a^\prime} $ if, and only if, $ a $ and $ a^\prime $ are congruent modulo $ m $.
\item
$ \varphi_0 = {\rm Id} $ and $ \varphi_a \circ \varphi_{a^\prime} = \varphi_{a^\prime} \circ \varphi_a = \varphi_{a + a^\prime} $. In particular, $ \varphi_a \circ \varphi_{n-a} = \varphi_{n-a} \circ \varphi_a = {\rm Id} $.
\item
For every $ q^r $-polynomial $ F(x) $, it holds that $ {\rm wt_R}(F) = {\rm wt_R}(\varphi_a(F)) $ (see (\ref{definition weight polynomial})), that is, $ \varphi_a $ is a rank equivalence.
\item
$ \varphi_a $ maps left ideals to left ideals and, in general, maps $ q^r $-cyclic codes to $ q^r $-cyclic codes.
\item
$ \varphi_a $ maps idempotents to idempotents.
\end{enumerate}
\end{proposition}
\begin{proof}
The first three items are straightforward calculations. The last two items follow from these first three items.

Finally, if $ \mathbf{c} = (c_0, c_1, \ldots, c_{n-1}) \in \mathbb{F}_{q^m}^n $, then the dimension of the vector space generated by $ c_0,c_1, \ldots, c_{n-1} $ in $ \mathbb{F}_{q^m} $ is the same as the dimension (over $ \mathbb{F}_q $) of the vector space generated by $ c_0^q, c_1^q, \ldots, c_{n-1}^q $, since rising to the power $ q $ is an $ \mathbb{F}_q $-linear automorphism of $ \mathbb{F}_{q^m} $. Therefore, $ {\rm wt_R}(c_0,c_1, \ldots, c_{n-1}) = {\rm wt_R}(c_0^q, c_1^q, \ldots, c_{n-1}^q) $. 

Since $ \varphi_a $ corresponds to rising to the power $ q^{rn-a} $, we see that it also preserves rank weights, and item 4 follows. 
\end{proof}

\begin{remark}
By item 6 in the previous proposition and Theorem \ref{roots}, item 5, we may obtain the minimal generator of a $ q^r $-cyclic code equivalent to a given one if we know the minimal generator or an idempotent of this latter code.
\end{remark}

On the other hand, these are the only maps coming from ring automorphisms of $ \mathcal{L}_{q^r} \mathbb{F}_{q^m}[x] / (x^{[rn]}-x) $ having the following reasonable properties: they commute with the $ q^r $-shifting operators (\ref{shifting operator}), are $ \mathbb{F}_q $-linear and leave the field $ \mathbb{F}_{q^m} $ invariant ($ \mathbb{F}_{q^m} $ is a subring of $ \mathcal{L}_{q^r} \mathbb{F}_{q^m}[x] / (x^{[rn]}-x) $ by considering any $ \alpha \in \mathbb{F}_{q^m} $ as the polynomial $ \alpha x $).

\begin{proposition}
For $ a = 0,1,2, \ldots, rn-1 $, if we view $ \varphi_a $ as a map $ \varphi_a : \mathbb{F}_{q^m}^n \longrightarrow \mathbb{F}_{q^m}^n $, then it holds that
$$ \sigma_{r,n} \circ \varphi_a = \varphi_a \circ \sigma_{r,n}, $$
where $ \sigma_{r,n} $ is as in (\ref{shifting operator}). Moreover, if $ \varphi $ is an $ \mathbb{F}_q $-linear ring automorphism of $ \mathcal{L}_{q^r} \mathbb{F}_{q^m}[x] / $ $ (x^{[rn]}-x) $ satisfying this condition and leaving $ \mathbb{F}_{q^m} $ invariant, then $ \varphi = \varphi_a $ for some $ a= 0,1,2, \ldots, rn-1 $.
\end{proposition}
\begin{proof}
The fact that a ring automorphism $ \varphi $ commutes with $ \sigma_{r,n} $ is equivalent to the condition
\begin{equation} \label{condition commutes}
 \varphi (x^{[1]} \otimes F) = x^{[1]} \otimes \varphi(F),
\end{equation}
for all $ F \in \mathcal{L}_{q^r} \mathbb{F}_{q^m}[x] / (x^{[rn]}-x) $, which is satisfied if $ \varphi = \varphi_a $. 

On the other hand, since $ \varphi (\alpha x + \beta x) = \varphi (\alpha x) + \varphi(\beta x) $ and $ \varphi (\alpha x \otimes \beta x) = \varphi (\alpha x) \otimes \varphi(\beta x) $, for all $ \alpha, \beta \in \mathbb{F}_{q^m} $, we have that $ \varphi $ is an automorphism of the field $ \mathbb{F}_{q^m} $ when restricted to constant polynomials $ \alpha x $. 

Moreover, if $ \alpha \in \mathbb{F}_q $, by $ \mathbb{F}_q $-linearity it holds that $ \varphi (\alpha x) = \alpha x \otimes \varphi(x) = \alpha x $. Hence $ \mathbb{F}_q $ is fixed by the automorphism induced by $ \varphi $ in $ \mathbb{F}_{q^m} $. Therefore, there exists an $ a = 0,1,2, \ldots, m-1 $ such that $ \varphi (\alpha x) = \alpha^{[nr - a]} x $, for all $ \alpha \in \mathbb{F}_{q^m} $. This together with (\ref{condition commutes}) means that $ \varphi = \varphi_a $ and we are done.
\end{proof}

Finally, we see that the lattice automorphism induced by $ \varphi_a $ in the lattice of $ q^r $-spaces over $ \mathbb{F}_{q^m} $ corresponds to the one induced by the field automorphism of $ \mathbb{F}_{q^{rn}} $ given by $ \beta \mapsto \beta^{[a]} $. In particular, by item 2 in Proposition \ref{maps a}, two of these automorphisms of the lattice of $ q^r $-root spaces over $ \mathbb{F}_{q^m} $, for $ a $ and $ a^\prime $, respectively, are equal if, and only if, $ a $ and $ a^\prime $ are congruent modulo $ m $. In short: 

\begin{proposition}
For all $ a = 0,1,2, \ldots, nr -1 $ and all $ F \in \mathcal{L}_{q^r} \mathbb{F}_{q^m}[x] / (x^{[rn]}-x) $, it holds that $ Z(\varphi_a(F)) = Z(F)^{[a]} $. In particular, $ Z(F)^{[a]} = Z(F)^{[a^\prime]} $ if $ a $ and $ a^\prime $ are congruent modulo $ m $. 
\end{proposition}

\section{Bounds on the minimum rank distance} \label{section bounds}

In this section we will give lower bounds on the minimum rank distance of $ q^r $-cyclic codes. The simplest bound on the minimum Hamming distance of classical cyclic codes is the BCH bound, which has been adapted to a bound on the minimum rank distance of $ q^r $-cyclic codes in \cite[Proposition 1]{skewcyclic3}. In this section, we will give two extensions of this bound analogous to the Hartmann-Tzeng bound \cite{hartmann} in the form of \cite[Theorem 2]{vanlint}, and another one analogous to the bound in \cite[Theorem 11]{vanlint}, also known as the shift bound.

\subsection{The rank-shift and rank-Hartmann-Tzeng bounds}

We start by giving the definition of independent sequence of $ \mathbb{F}_{q^r} $-linear vector subspaces of $ \mathbb{F}_{q^{rn}} $ with respect to some $ \mathbb{F}_{q^r} $-linear subspace $ S \subseteq \mathbb{F}_{q^{rn}} $.

\begin{definition} \label{def independent}
Given $ \mathbb{F}_{q^r} $-linear subspaces $ S, I_0, I_1, I_2, \ldots \subseteq \mathbb{F}_{q^{rn}} $, we say that the sequence $ I_0, I_1, I_2, \ldots $ is independent with respect to $ S $ if the following hold:
\begin{enumerate}
\item
$ I_0 = \{ \mathbf{0} \} $.
\item
For $ i > 0 $, either 
\begin{enumerate}
\item
$ I_i = I_j \oplus \langle \beta \rangle $, for some $ 0 \leq j < i $, $ I_j \subseteq S $ and $ \beta \notin S $, or
\item
$ I_i = I_j^{[br]} $, for some $ 0 \leq j < i $ and some integer $ b \geq 0 $.
\end{enumerate}
\end{enumerate}
We say that a subspace $ I \subseteq \mathbb{F}_{q^{rn}} $ is independent with respect to $ S $ if it is a space in a sequence that is independent with respect to $ S $.
\end{definition}

The van Lint-Wilson or shift bound \cite[Theorem 11]{vanlint} for the rank metric becomes then as follows. Observe that it is a bound on the rank weight (see (\ref{definition weight polynomial})) of a given $ q^r $-polynomial in $ \mathcal{L}_{q^r} \mathbb{F}_{q^m}[x] / (x^{[rn]} - x) $ in terms of its roots.

\begin{theorem} [\textbf{Rank-shift bound}] \label{van lint}
Let $ F \in \mathcal{L}_{q^r} \mathbb{F}_{q^m}[x] / (x^{[rn]} - x) $ and $ S = Z(F) = \{ \beta \in \mathbb{F}_{q^{rn}} \mid F(\beta) = 0 \} $, as in Definition \ref{def root space}. If $ I \subseteq \mathbb{F}_{q^{rn}} $ is an $ \mathbb{F}_{q^r} $-linear subspace independent with respect to $ S $, then
$$ {\rm wt_R}(F) \geq \dim_{\mathbb{F}_{q^r}}(I), $$
where $ {\rm wt_R}(F) $ is as in (\ref{definition weight polynomial}).
\end{theorem}
\begin{proof}
Define the vector $ \mathbf{F} = (F_0, F_1, \ldots, F_{n-1}) \in \mathbb{F}_{q^m}^n $ if $ F = F_0 x + F_1 x^{[r]} + \cdots + F_{n-1} x^{[(n-1)r]} $ (recall (\ref{definition gamma})). Now write $ \mathbf{F} = \sum_{i=0}^{m-1} \alpha_i \mathbf{F}_i $, where $ \mathbf{F}_i \in \mathbb{F}_q^n $, for $ i = 0,1, \ldots, m-1 $ and $ \alpha_0, \alpha_1, \ldots, \alpha_{m-1} $ is a basis of $ \mathbb{F}_{q^m} $ over $ \mathbb{F}_q $. Define $ w = {\rm wt_R}(F) $, and recall from Subsection \ref{subsec rank codes} that $ w = \dim_{\mathbb{F}_q}(\langle \mathbf{F}_0, \mathbf{F}_1, \ldots, \mathbf{F}_{m-1} \rangle_{\mathbb{F}_q}) $.

Let $ A $ be a $ w \times n $ matrix over $ \mathbb{F}_q $ whose rows generate the vector space $ \langle \mathbf{F}_0, \mathbf{F}_1, \ldots, $ $ \mathbf{F}_{m-1} \rangle_{\mathbb{F}_q} $. Since $ A $ is full-rank, there exists a $ w \times n $ matrix $ A^\prime $ over $ \mathbb{F}_q $ such that $ A A^{\prime T} = I $. On the other hand, by definition of $ A $, there exist $ \mathbf{x}_i \in \mathbb{F}_q^w $ with $ \mathbf{F}_i = \mathbf{x}_i A $, for $ i = 0,1, \ldots, m-1 $. It follows that 
$$ \mathbf{F} (A^{\prime T} A) = \sum_{i=0}^{m-1} \alpha_i \mathbf{x}_i A (A^{\prime T} A) = \sum_{i=0}^{m-1} \alpha_i \mathbf{x}_i (A A^{\prime T}) A = \sum_{i=0}^{m-1} \alpha_i \mathbf{x}_i A = \mathbf{F}. $$

On the other hand, for an $ \mathbb{F}_{q^r} $-linear subspace $ J \subseteq \mathbb{F}_{q^{rn}} $, define the $ \mathbb{F}_{q^{rn}} $-linear subspace of $ \mathbb{F}_{q^{rn}}^w $ given by
$$ V(J) = \langle \{ (\beta, \beta^{[r]}, \beta^{[2r]}, \ldots, \beta^{[(n-1)r]}) A^T \mid \beta \in J \} \rangle_{\mathbb{F}_{q^{rn}}} \subseteq \mathbb{F}_{q^{rn}}^w. $$
We will prove that $ \dim_{\mathbb{F}_{q^{rn}}}(V(I)) = \dim_{\mathbb{F}_{q^r}}(I) $, and hence it will follow that $ w \geq \dim_{\mathbb{F}_{q^r}}(I) $.

By definition, there exists a sequence $ I_0, I_1, I_2, \ldots \subseteq \mathbb{F}_{q^{rn}} $ of $ \mathbb{F}_{q^r} $-linear subspaces that is independent with respect to $ S $ and $ I = I_i $, for some $ i $. We will prove by induction on $ i $ that $ \dim_{\mathbb{F}_{q^{rn}}}(V(I_i)) = \dim_{\mathbb{F}_{q^r}}(I_i) $.

For $ i = 0 $, we have that $ I_0 = \{ 0 \} $ and $ V(I_0) = \{ \mathbf{0} \} $, and the statement is true.

Fix $ i > 0 $ and assume that it is true for all $ 0 \leq j < i $. The space $ I_i $ may be obtained in two different ways, according to Definition \ref{def independent}:

First, assume that $ I_i = I_j \oplus \langle \beta \rangle $, with $ 0 \leq j < i $, $ I_j \subseteq S $ and $ \beta \notin S $. Therefore, $ \dim_{\mathbb{F}_{q^r}}(I_i) = \dim_{\mathbb{F}_{q^r}}(I_j) + 1 $. It follows that $ \dim_{\mathbb{F}_{q^{rn}}}(V(I_i)) \leq \dim_{\mathbb{F}_{q^{rn}}}(V(I_j)) + 1 $. Assume that $ \dim_{\mathbb{F}_{q^{rn}}}(V(I_i)) = \dim_{\mathbb{F}_{q^{rn}}}(V(I_j)) $. This means that 
$$ (\beta, \beta^{[r]}, \beta^{[2r]}, \ldots, \beta^{[(n-1)r]}) A^T \in V(I_j). $$
On the other hand, for every $ \gamma \in S $, it holds that
$$ 0 = F(\gamma) = \mathbf{F} (\gamma, \gamma^{[r]}, \ldots, \gamma^{[(n-1)r]})^T = (\mathbf{F}A^{\prime T}) (A (\gamma, \gamma^{[r]}, \ldots, \gamma^{[(n-1)r]})^T). $$
Since $ (\beta, \beta^{[r]}, \beta^{[2r]}, \ldots, \beta^{[(n-1)r]}) A^T $ is a linear combination (over $ \mathbb{F}_{q^{rn}} $) of vectors in $ V(I_j) $, it follows that
$$ 0 = (\mathbf{F}A^{\prime T}) (A (\beta, \beta^{[r]}, \ldots, \beta^{[(n-1)r]})^T) = \mathbf{F} (\beta, \beta^{[r]}, \ldots, \beta^{[(n-1)r]})^T = F(\beta), $$
which means that $ \beta \in S $, a contradiction. Thus $ \dim_{\mathbb{F}_{q^{rn}}}(V(I_i)) = \dim_{\mathbb{F}_{q^{rn}}}(V(I_j)) + 1 $ and the result holds in this case.

Now assume that $ I_i = I_j^{[br]} $, for some integer $ b \geq 0 $ and $ 0 \leq j < i $. Since rising to the power $ q^r $ in $ \mathbb{F}_{q^{rn}} $ is an $ \mathbb{F}_{q^r} $-linear vector space automorphism, we have that  $ \dim_{\mathbb{F}_{q^r}}(I_i) = \dim_{\mathbb{F}_{q^r}}(I_j) $. On the other hand, rising to the power $ q^r $ in $ \mathbb{F}_{q^{rn}}^w $ is an $ \mathbb{F}_{q^{rn}} $-semilinear vector space automorphism, which also preserve dimensions over $ \mathbb{F}_{q^{rn}} $. Since $ V(I_i) = V(I_j)^{[br]} $, we have that $ \dim_{\mathbb{F}_{q^{rn}}}(V(I_i)) = \dim_{\mathbb{F}_{q^{rn}}}(V(I_j)) $ and the result holds also in this case.
\end{proof}

Future research on other possible generalizations of the rank-BCH bound could be trying to obtain rank versions of the bounds in \cite{duursmasymmetric,pellikaanshift}, to cite some. We next give a toy example to illustrate the previous bound:

\begin{example}
Let $ r=1 $, $ n = m = 2 $. Take a vector $ \mathbf{F} = (F_0, F_1) \in \mathbb{F}_{q^2}^2 $. We next see that the previous bound gives the exact value of $ {\rm wt_R}(\mathbf{F}) $. Observe that $ {\rm wt_R}(\gamma \mathbf{F}) = {\rm wt_R}(\mathbf{F}) $, for all non-zero $ \gamma \in \mathbb{F}_{q^2} $, and hence we may assume $ \mathbf{F} = (1, \alpha) $ for some $ \alpha \in \mathbb{F}_{q^2} $. Let $ S = Z(F) \subseteq \mathbb{F}_{q^2} $, for $ F(x) = x^{[1]} + \alpha x $, and distinguish two cases:
\begin{enumerate}
\item
$ {\rm wt_R}(\mathbf{F}) = 1 $, that is, $ \alpha \in \mathbb{F}_q $: We have that $ S = \{ \mathbf{0} \} $ if $ \alpha = 0 $, and  $ S = \langle \beta \rangle $, for some non-zero $ \beta \in \mathbb{F}_{q^2} $ if $ \alpha \neq 0 $. We may start constructing an independent sequence by $ I_1 = \langle \gamma \rangle $, for some  $ \gamma \in \mathbb{F}_{q^2} \setminus S $. We see that these (and $ I_0 = \{ 0 \} $) are all subspaces independent with respect to $ S $, and hence we may only construct an independent space of dimension $ 1 $.
\item
$ {\rm wt_R}(\mathbf{F}) = 2 $, that is, $ \alpha \in \mathbb{F}_{q^2} \setminus \mathbb{F}_q $: In this case, $ S = \langle \beta \rangle $, for some $ \beta \in \mathbb{F}_{q^2} $. Then $ \beta $ and $ \beta^q $ are linearly independent over $ \mathbb{F}_q $ since $ \beta^q + \alpha \beta = 0 $ and $ \alpha \notin \mathbb{F}_q $.

Define $ I_1 = \langle \beta^q \rangle $, then $ I_2 = I_1^q = \langle \beta \rangle $ and finally $ I_3 = I_2 \oplus \langle \beta^q \rangle = \langle \beta, \beta^q \rangle $. It holds that $ \dim(I_3) = 2 $, hence the previous bound is an equality: $ 2 = {\rm wt_R}(\mathbf{F}) \geq \dim(I_3) = 2 $.
\end{enumerate}
\end{example}

As a consequence of the previous theorem, we may give the following bound, analogous to the Hartmann-Tzeng bound as it appears in \cite[Theorem 2]{vanlint}:

\begin{corollary} [\textbf{Rank-HT bound}] \label{hartmann bound}
Take integers $ c > 0 $, $ \delta > 0 $ and $ s \geq 0 $, with $ \delta + s \leq \min \{ m,n \} $ and $ d = {\rm gcd}(c,n) < \delta $, and let $ \alpha \in \mathbb{F}_{q^{rn}} $ be such that $ A = \{ \alpha^{[(i + jc)r]} \mid 0 \leq i \leq \delta -2, 0 \leq j \leq s \} $ is a linearly independent (over $ \mathbb{F}_{q^r} $) set of vectors, not necessarily pairwise distinct. 

If $ F \in \mathcal{L}_{q^r} \mathbb{F}_{q^m}[x] / (x^{[rn]}-x) $ satisfies that $ A \subseteq T = Z(F) $, then $ {\rm wt_R}(F) \geq \delta + s $ (recall (\ref{definition weight polynomial})). In particular, if $ C = \rho_r^{-1}(T) $, with $ \rho_r $ as in Definition \ref{def root space}, then 
$$ d_R(C) \geq \delta + s. $$
\end{corollary}
\begin{proof}
First, since $ \delta + s \leq n $, we have that $ d s < \delta s \leq n $, and $ n/d $ is the order of $ c $ modulo $ n $. Hence, the elements $ jcr $, for $ j = 0,1,2, \ldots, s $, are all distinct modulo $ rn $. 

On the other hand, we may assume that $ A $ is maximal with the given structure. That is, there exists $ 0 \leq i \leq \delta-2 $ with $ \alpha^{[(i + (s+1)c)r]} \notin T $ and there exists $ 0 \leq j \leq s $ such that $ \alpha^{[(\delta - 1 + jc)r]} \notin T $. From the proof, we will see that we may assume for simplicity that $ j = 0 $, and by repeatedly raising to the power $ q^{r} $, we will also see that we may assume that $ i = \delta - 2 $.

We will now define a suitable sequence $ I_0,I_1, I_2, \ldots \subseteq \mathbb{F}_{q^{rn}} $ of $ \mathbb{F}_{q^r} $-linear spaces independent with respect to $ S = T $, and with $ \dim_{\mathbb{F}_{q^r}}(I_i) \geq \delta + s $ for some $ i \geq 0 $. We start by $ I_0 = \{ \mathbf{0} \} $, and $ I_{2i+1} = I_{2i} \oplus \langle \alpha^{[(\delta - 2 + (s+1)c)r]} \rangle $ and $ I_{2i+2} = I_{2i+1}^{[(n - c)r]} $, for $ i = 0,1,2, \ldots, s $. 

We see by induction that $ J_1 = I_{2s + 2} $ is generated by the set 
$$ \{ \alpha^{[(\delta - 2 + jc)r]} \mid 0 \leq j \leq s \}. $$
Next, define $ J_{2i + 1} = J_{2i} \oplus \langle \alpha^{[(\delta - 1)r]} \rangle $ and $ J_{2i} = J_{2i-1}^{[(n-1)r]} $, for $ i = 1,2, \ldots, \delta - 1 $. 

Finally, again by induction we see that $ J_{2\delta - 1} $ is generated by the set 
\begin{equation} \label{sets}
\{ \alpha^{[ir]} \mid 0 \leq i \leq \delta - 1 \} \cup \{ \alpha^{[jcr]} \mid 1 \leq j \leq s \},
\end{equation}
whose elements are all distinct by the first two paragraphs in the proof: First, these two sets are disjoint. If $ \alpha^{[jcr]} = \alpha^{[ir]} $, for some $ 1 \leq i \leq \delta - 1 $ and $ 1 \leq j \leq s $, then by considering $ jc, jc+1, \ldots, jc + \delta -2 $, we see that $ \alpha^{[(\delta - 1)r]} \in T $, a contradiction. Now, if two elements in the set on the left are equal, then we see again that $ \alpha^{[(\delta - 1)r]} \in T $. Finally, if two elements in the set on the right are equal, we may now see that $ \alpha^{[(\delta -2 + (s+1)c)r]} \in T $, which is again a contradiction.

Since there are $ \delta + s $ elements in the set (\ref{sets}) and they are linearly independent by hypothesis, the result follows from the previous theorem.
\end{proof}

By taking $ s=0 $ and $ c = 1 $, we see that the rank version of the BCH bound obtained in \cite[Proposition 1]{skewcyclic3} is a corollary of the previous bound:

\begin{corollary} [\textbf{Rank-BCH bound \cite[Proposition 1]{skewcyclic3}}] \label{rank BCH bound}
Take an integer $ \delta > 0 $, with $ \delta \leq \min \{ m,n \} $, and let $ \alpha \in \mathbb{F}_{q^{rn}} $ be such that $ \alpha, \alpha^{[r]}, \alpha^{[2r]}, \ldots, $ $ \alpha^{[(\delta - 2)r]} $ are linearly independent over  $ \mathbb{F}_{q^r} $. 

If $ F \in \mathcal{L}_{q^r} \mathbb{F}_{q^m}[x] / (x^{[rn]}-x) $ satisfies that $ T = Z(F) $ contains the previous elements, then $ {\rm wt_R}(F) \geq \delta $ (recall (\ref{definition weight polynomial})). In particular, if $ C = \rho_r^{-1}(T) $, with $ \rho_r $ as in Definition \ref{def root space}, then 
$$ d_R(C) \geq \delta. $$
\end{corollary}

Thanks to the lattice study of the previous two sections and, in particular, thanks to Proposition \ref{cyclotomic normal basis}, we can see that it is not difficult to find examples where the rank-HT bound beats the rank-BCH bound, as in the classical case:

\begin{example}
Consider $ r=1 $, $ n=2m $ and $ m = 31 $, and take a normal basis $ \alpha, \alpha^{[1]}, \ldots, $ $ \alpha^{[61]} $ of $ \mathbb{F}_{q^{62}} $ over $ \mathbb{F}_{q} $. Take $ c = 5 $, $ \delta = 4 $ and $ s = 3 $, and the $ q $-root space
$$ T = (C_q(\alpha) \oplus C_q(\alpha^{[1]}) \oplus C_q(\alpha^{[2]})) \oplus (C_q(\alpha^{[5]}) \oplus C_q(\alpha^{[6]}) \oplus C_q(\alpha^{[7]})) $$
$$ \oplus (C_q(\alpha^{[10]}) \oplus C_q(\alpha^{[11]}) \oplus C_q(\alpha^{[12]})) \oplus (C_q(\alpha^{[15]}) \oplus C_q(\alpha^{[16]}) \oplus C_q(\alpha^{[17]})). $$
By Proposition \ref{cyclotomic normal basis}, we have that $ C_q(\alpha^{[i]}) $ has $ \{ \alpha^{[i]}, \alpha^{[31 + i]} \} $ as a basis, and hence has dimension $ 2 $. Therefore, the code $ C = \rho_r^{-1}(T) $ has dimension $ 62 - 24 = 38 $. The rank-BCH bound states that $ d_R(C) \geq 4 $, whereas the rank-HT bound improves it giving $ d_R(C) \geq 7 $. 
\end{example}

\subsection{Rank-BCH codes from normal bases are generalized Gabidulin codes}

As a consequence of the bound in Corollary \ref{rank BCH bound}, a family of $ q^r $-cyclic codes with a designed minimum rank distance is defined in \cite[Section 3]{skewcyclic3}, in analogy with classical BCH codes. By means of difference equations and Casoratian determinants, rank-BCH codes are defined in \cite{skewcyclic3} as $ q^r $-cyclic codes with prescribed minimum rank distance and generator polynomial of minimal degree. 

We will give an alternative description in terms of $ q^r $-cyclotomic spaces, which will allow us to prove that, when $ m=n $ and $ r $ and $ n $ are coprime, rank-BCH codes from normal bases are generalized Gabidulin codes, as in Subsection \ref{subsec rank codes}, which are MRD. 

\begin{definition} \label{def BCH}
Given $ 1 \leq \delta \leq m $, we say that the $ q^r $-cyclic code $ C(x) $ over $ \mathbb{F}_{q^m} $ is a rank-BCH code of designed minimum rank distance $ \delta $ if the corresponding $ q^r $-root space $ T $ over $ \mathbb{F}_{q^m} $ (see Definition \ref{def root space}) is
$$ T = C_{q^r}(\alpha) + C_{q^r}(\alpha^{[r]}) + C_{q^r}(\alpha^{[2r]}) + \cdots + C_{q^r}(\alpha^{[(\delta-2)r]}), $$
where $ \alpha \in \mathbb{F}_{q^{rn}} $ and $ \alpha, \alpha^{[r]}, \alpha^{[2r]}, \ldots , \alpha^{[(\delta-2)r]} $ are linearly independent over $ \mathbb{F}_{q^r} $.
\end{definition}

The following result follows immediately from Corollary \ref{rank BCH bound}:

\begin{proposition}
The rank-BCH code $ C(x) $ in the previous definition satisfies that
$$ d_R(C) \geq \delta. $$
\end{proposition}

If $ m=n $ and $ r $ and $ n $ are coprime, the Gabidulin codes $ {\rm Gab}_{k,r}(\boldsymbol\beta) $ defined using a normal basis (see Subsection \ref{subsec rank codes}) are rank-BCH codes also using normal bases, and viceversa, and all of them are MRD codes. Hence the family of rank-BCH codes include MRD codes. We will use \cite[Lemma 2]{new-construction}, which is the following:

\begin{lemma} [\textbf{\cite[Lemma 2]{new-construction}}]
If $ r $ and $ n $ are coprime and $ \alpha_0, \alpha_1, \ldots, \alpha_{n-1} \in \mathbb{F}_{q^n} $ are linearly independent over $ \mathbb{F}_q $, then they are also linearly independent over $ \mathbb{F}_{q^r} $, considered as elements in $ \mathbb{F}_{q^{rn}} $.
\end{lemma}

\begin{theorem}
Assume $ m=n $ and $ r $ and $ n $ are coprime. Take a normal basis $ \alpha, \alpha^{[1]}, \ldots, $ $ \alpha^{[n-1]} \in \mathbb{F}_{q^n} = \mathbb{F}_{q^m} $ and $ 1 \leq \delta \leq n $. Then the corresponding rank-BCH code $ C(x) $, as in Definition \ref{def BCH}, is the generalized Gabidulin code $ {\rm Gab}_{k,r}(\boldsymbol\alpha) $ (see Subsection \ref{subsec rank codes}), where $ \boldsymbol\alpha = (\alpha, \alpha^{[r]}, \ldots, \alpha^{[(n-1)r]}) $ and $ k = n - \delta + 1 $. 
\end{theorem}
\begin{proof}
Since $ m=n $, we have that $ \alpha \in \mathbb{F}_{q^m} $, and hence $ C_{q^r}(\alpha^{[i]}) = \langle \alpha^{[i]} \rangle_{\mathbb{F}_{q^r}} $, for all $ i=0,1,2, \ldots, n-1 $. Therefore, the $ q^r $-root space $ T $ corresponding to $ C(x) $ is $ T = \langle \alpha, \alpha^{[r]}, \ldots, \alpha^{[(\delta - 2)r]} \rangle_{\mathbb{F}_{q^r}} $, whose dimension over $ \mathbb{F}_{q^r} $ is $ \delta - 1 $ by the previous lemma.

Hence, by item 4 in Theorem \ref{roots}, the matrix $ \mathcal{M}(\alpha, \alpha^{[r]}, \ldots, \alpha^{[(\delta -2)r]}) $ is a parity check matrix of $ C $ over $ \mathbb{F}_{q^m} $. However, this is also the parity check matrix of the above mentioned Gabidulin code of dimension $ k $, $ \mathcal{H}_{k,r}(\boldsymbol\alpha) $, if $ k = n - \delta + 1 $ (see Subsection \ref{subsec rank codes}). Therefore both are equal and the theorem follows.
\end{proof}

\section{General $ \mathbb{F}_q $-linear skew cyclic codes: Connecting Hamming-metric cyclic codes and rank-metric skew cyclic codes} \label{general cyclic}

To conclude, we will give some first steps in the general study of $ \mathbb{F}_q $-linear $ q^r $-cyclic codes in $ \mathbb{F}_{q^m}^n $. 

Its main interest for our purposes is that they include both the family of skew cyclic codes in the rank metric, which are the main topic of this paper, and the classical family of cyclic codes in the Hamming metric, as we will prove in the first subsection. 

Moreover, as we will see in the second subsection, some $ \mathbb{F}_{q^m} $-linear $ q^r $-cyclic codes in the rank metric with $ m=n $ actually are obtained from cyclic codes in the Hamming metric via $ \mathbb{F}_q $-linear $ q^r $-cyclic codes, which will allow us to compare their parameters and give a negative criterion of MRD skew cyclic codes in terms of MDS cyclic codes.

\subsection{Hamming-metric cyclic codes are rank-metric skew cyclic codes}

Assume in this subsection that $ m=n $, fix a basis $ \alpha_0, \alpha_1, \ldots, \alpha_{n-1} $ of $ \mathbb{F}_{q^n} $ over $ \mathbb{F}_q $ and define the map $ E : \mathbb{F}_q^n \longrightarrow \mathbb{F}_{q^n}^n $ by 
\begin{equation} \label{map E}
E(c_0,c_1, \ldots, c_{n-1}) = (c_0 \alpha_0, c_1 \alpha_1, \ldots, c_{n-1} \alpha_{n-1}).
\end{equation}
This map is one to one, $ \mathbb{F}_q $-linear and $ {\rm wt_H}(\mathbf{c}) = {\rm wt_R}(E(\mathbf{c})) $, where $ {\rm wt_H}(\mathbf{c}) $ denotes the Hamming weight of the vector $ \mathbf{c} $. Therefore, the codes $ C \subseteq \mathbb{F}_q^n $ and $ E(C) \subseteq \mathbb{F}_{q^n}^n $ behave equally, where we consider the Hamming metric for $ C $ and the rank metric for $ E(C) $.

Assume also in this subsection that $ n $ and $ r $ are coprime and $ \alpha_0, \alpha_1, \ldots, \alpha_{n-1} $ satisfies that $ \alpha_i = \alpha^{[ir]} $, for $ i = 0,1,2, \ldots, n-1 $, where $ \alpha, \alpha^{[1]}, \ldots, \alpha^{[n-1]} $ is a normal basis. In this case, classical cyclic codes correspond to $ q^r $-cyclic codes.

\begin{theorem} \label{classical skew}
With the assumptions as in the previous paragraph, an arbitrary (linear or non-linear) code $ C \subseteq \mathbb{F}_q^n $ is cyclic if, and only if, the code $ E(C) \subseteq \mathbb{F}_{q^n}^n $ is $ q^r $-cyclic. 

Moreover, $ C $ is $ \mathbb{F}_q $-linear if, and only if, so is $ E(C) $, and the Hamming-metric behaviour of $ C $ is the same as the rank-metric behaviour of $ E(C) $, since $ {\rm wt_H}(\mathbf{c}) = {\rm wt_R}(E(\mathbf{c})) $, for all $ \mathbf{c} \in \mathbb{F}_q^n $.
\end{theorem}
\begin{proof}
Let $ \mathbf{c} = (c_0, c_1, \ldots, c_{n-1}) \in C $ and $ E(\mathbf{c}) = (d_0, d_1, \ldots, d_{n-1}) \in E(C) $. Then 
$$ E(c_{n-1}, c_0, c_1, \ldots, c_{n-2}) = (c_{n-1} \alpha, c_0 \alpha^{[r]}, \ldots, c_{n-2} \alpha^{[(n-1)r]}) $$
$$ =((c_{n-1} \alpha^{[(n-1)r]})^{q^r}, (c_0 \alpha)^{q^r}, \ldots, (c_{n-2} \alpha^{[(n-2)r]})^{q^r}) = (d_{n-1}^{q^r}, d_0^{q^r}, \ldots, d_{n-2}^{q^r}), $$
and the result follows, since the linearity claim is trivial from the linearity of $ E $.
\end{proof}

\subsection{MRD skew cyclic codes and MDS cyclic codes}

We will now relate MRD $ \mathbb{F}_{q^n} $-linear $ q^r $-cyclic codes in $ \mathbb{F}_{q^n}^n $ with classical MDS $ \mathbb{F}_q $-linear cyclic codes in $ \mathbb{F}_q^n $. We first need some properties of $ \mathbb{F}_q $-linear $ q^r $-cyclic codes. The following lemma is proven in the same way as Lemma \ref{char qcyclic}:

\begin{lemma} \label{char general}
A code $ C \subseteq \mathbb{F}_{q^m}^n $ is $ \mathbb{F}_q $-linear and $ q^r $-cyclic if, and only if, $ C(x) $ satisfies that $ G - H \in C(x) $ and $ F \otimes G \in C(x) $, for all $ F(x) \in \mathcal{L}_{q^r}\mathbb{F}_{q}[x] $ and all $ G,H \in C(x) $.
\end{lemma}

\begin{definition}
A subset $ C(x) \subseteq \mathcal{L}_{q^r} \mathbb{F}_{q^m}[x] / (x^{[rn]} - x) $ satisfying the conditions in the previous lemma is called an $ \mathbb{F}_q $-left ideal.
\end{definition}

By Theorem \ref{classical skew} and Lemma \ref{char general}, classical cyclic codes for the Hamming metric can be seen as $ \mathbb{F}_q $-left ideals in $ \mathcal{L}_{q^r} \mathbb{F}_{q^n}[x] / (x^{[rn]} - x) $ for the rank metric, provided that $ n $ and $ r $ are coprime.

We observe that $ \mathbb{F}_q $-left ideals are finitely generated. That is, every $ \mathbb{F}_q $-left ideal is of the form $ C(x) = (G_1,G_2, \ldots, G_t)_{\mathbb{F}_q} $, where we define
$$ (G_1,G_2, \ldots, G_t)_{\mathbb{F}_q} = \left\lbrace \sum_{i=1}^t Q_i \otimes G_i \mid Q_i(x) \in \mathcal{L}_{q^r} \mathbb{F}_q[x] \right\rbrace. $$
However, not all $ \mathbb{F}_q $-left ideals are principal, that is, of the form $ (G)_{\mathbb{F}_q} $, for some $ G(x) \in \mathcal{L}_{q^r}\mathbb{F}_{q^m}[x] $. The following proposition relates the dimension of an $ \mathbb{F}_q $-left ideal and its number of generators. We also describe generators of the vector space $ C $ over $ \mathbb{F}_q $ as in Theorem \ref{generator}:

\begin{proposition} \label{generators general}
Let $ C(x) $ be an $ \mathbb{F}_q $-left ideal with $ C(x) = (G_1,G_2, \ldots, G_t)_{\mathbb{F}_q} $. It holds that:
\begin{enumerate}
\item
$ C(x) $ is generated by $ x^{[j]} \otimes G_i $ as an $ \mathbb{F}_q $-linear vector space, for $ j = 0,1, \ldots, n-1 $ and $ i = 1,2, \ldots, t $. In particular, a basis of $ C $ over $ \mathbb{F}_q $ may be obtained from the set of vectors 
$$ (G_{i,n-j}^{[jr]}, G_{i,n-j+1}^{[jr]}, \ldots, G_{i,n-j-1}^{[jr]}), $$
for the previous $ i $ and $ j $, where $ G_i(x) = G_{i,0} x + G_{i,1} x^{[r]} + \cdots + G_{i,n-1} x^{[n-1]} $.
\item
The dimension of $ C $ (over $ \mathbb{F}_q $) satisfies $ \dim(C(x)) \leq tn $.
\item
There exist $ F_1, F_2, \ldots, F_{mn} \in C(x) $ such that $ C(x) = (F_1, F_2, \ldots, F_{mn})_{\mathbb{F}_q} $. 
\end{enumerate}
\end{proposition}
\begin{proof}
The first item follows from the fact that $ x^{[j]} \otimes G_j $ corresponds to the vector $ (G_{i,n-j}^{[jr]}, G_{i,n-j+1}^{[jr]}, \ldots, G_{i,n-j-1}^{[jr]}) $. The second item follows from this first item, and the third item follows from the fact that $ \dim(C) \leq mn $.
\end{proof}

Now we see that classical cyclic codes actually correspond to principal $ \mathbb{F}_q $-left ideals. For that purpose, let the assumptions be as in Theorem \ref{classical skew} and define the operators $ L,E : \mathbb{F}_q[x] / (x^n - 1) \longrightarrow \mathcal{L}_{q^r} \mathbb{F}_{q^n}[x]/(x^{[rn]} - x) $ as
$$ L (f_0 + f_1 x + \cdots + f_{n-1} x^{n-1}) = f_0x + f_1 x^{[r]} + \cdots + f_{n-1} x^{[(n-1)r]}, \textrm{ and} $$
$$ E (g_0 + g_1 x + \cdots + g_{n-1} x^{n-1}) = g_0 \alpha x + g_1 \alpha^{[r]} x^{[r]} + \cdots + g_{n-1} \alpha^{[(n-1)r]} x^{[(n-1)r]}, $$
where $ f_i,g_i \in \mathbb{F}_q $, for $ i = 0,1, \ldots, n-1 $.
 
\begin{proposition} \label{principal classical}
With the assumptions as in Theorem \ref{classical skew}, for all $ f(x), g(x) \in \mathbb{F}_q[x] $ $ /(x^n - 1) $, it holds that
\begin{equation} \label{eq 1 in principal}
L(f(x)) \otimes E(g(x)) = E(f(x) g(x)).
\end{equation}
In particular, if $ [ g(x) ] $ denotes the ideal in $ \mathbb{F}_q[x]/(x^n - 1) $ generated by $ g(x) $, then 
\begin{equation} \label{eq 2 in principal}
E([g(x)]) = (E(g(x)))_{\mathbb{F}_q}.
\end{equation}
This means that, if $ C \subseteq \mathbb{F}_q^n $ is cyclic, then $ E(C)(x) $ is a principal $ \mathbb{F}_q $-left ideal generated by $ E(g(x)) $ if $ g(x) $ generates the ideal in $ \mathbb{F}_q[x]/(x^n - 1) $ corresponding to $ C $.
\end{proposition}
\begin{proof}
If $ f(x) = f_0 + f_1 x + \cdots + f_{n-1} x^{n-1} $ and $ g(x) = g_0 + g_1 x + \cdots + g_{n-1} x^{n-1} $, then
$$ L(f(x)) \otimes E(g(x)) = \sum_{i=0}^{n-1} \left( \sum_{j=0}^{n-1} f_{i-j} g_j (\alpha^{[jr]})^{[(i-j)r]} \right) x^{[ir]} $$
$$ = \sum_{i=0}^{n-1} \left( \sum_{j=0}^{n-1} f_{i-j} g_j \right) \alpha^{[ir]} x^{[ir]} = E(f(x) g(x)), $$
and Equation (\ref{eq 1 in principal}) follows. The second part (\ref{eq 2 in principal}) follows immediately from (\ref{eq 1 in principal}).
\end{proof}

On the other hand, if $ C(x) = (G_1,G_2, \ldots, G_t)_{\mathbb{F}_q} $, then the $ \mathbb{F}_{q^m} $-linear code generated by $ C(x) $ is
$$ C(x)_{\mathbb{F}_{q^m}} = (G_1,G_2, \ldots, G_t) = (D), $$
where $ D $ is the greatest common divisor of $ G_1, G_2, \ldots, G_t $ in the quotient ring $ \mathcal{L}_{q^r} \mathbb{F}_{q^m}[x] / $ $ (x^{[rn]} - x) $. Therefore, $ d_R(C(x)) \geq d_R( (D) ) $, and the $ q^r $-root space $ T = Z(D) = Z(G_1) \cap Z(G_2) \cap \ldots \cap Z(G_t) $ may be used to give bounds on the minimum rank distance of $ C(x) $, using for example the bounds in Section \ref{section bounds}.

Now we come to the main result in this subsection, where we see that the $ \mathbb{F}_{q^n} $-linear code generated by a classical cyclic code is again principal, with the same minimal generator and corresponding dimension, but its minimum rank distance is lower than the minimum Hamming distance of the original cyclic code. In particular, this gives a negative criterion for MRD skew cyclic codes in terms of MDS cyclic codes. 

\begin{theorem} \label{classical minimal}
With the assumptions as in Theorem \ref{classical skew}, if $ g(x) \in \mathbb{F}_q[x] /(x^n - 1) $ is the minimal generator of the $ \mathbb{F}_q $-linear cyclic code $ C \subseteq \mathbb{F}_q^n $ and $ \widehat{C} = \langle E(C) \rangle_{\mathbb{F}_{q^n}} $, then $ \widehat{C} $ is the $ \mathbb{F}_{q^n} $-linear $ q^r $-cyclic code corresponding to
$$ \widehat{C}(x) = (E(g(x))). $$
Moreover, $ E(g(x)) $ is the minimal generator of $ \widehat{C}(x) $, and:
\begin{enumerate}
\item
$ d_R(\widehat{C}) \leq d_H(C) $, $ \dim_{\mathbb{F}_{q^n}}(\widehat{C}) = \dim_{\mathbb{F}_q}(C) $.
\item
If $ \widehat{C} $ is MRD, then $ C $ is MDS.
\end{enumerate}
\end{theorem}
\begin{proof}
It is well-known that the shifted vectors in $ \mathbb{F}_q^n $, 
$$ (g_0, g_1, \ldots, g_{n-k}, 0, \ldots, 0), (0,g_0, g_1, \ldots, g_{n-k}, 0, \ldots, 0), \ldots, $$
$$ (0, \ldots, 0, g_0, g_1, \ldots, g_{n-k}) $$
constitute a basis of $ C $, where $ g(x) = g_0 + g_1x + \cdots + g_{n-k} x^{n-k} $ and $ g_{n-k} \neq 0 $. By Proposition \ref{generators general} and Proposition \ref{principal classical}, the $ q^r $-shifted vectors in $ \mathbb{F}_{q^n}^n $, 
$$ (g_0 \alpha, g_1 \alpha^{[r]}, \ldots, g_{n-k} \alpha^{[(n-k)r]}, 0, \ldots, 0), $$
$$ (0, g_0 \alpha^{[r]}, g_1 \alpha^{[2r]}, \ldots, g_{n-k} \alpha^{[(n-k+1)r]}, 0, \ldots, 0), \ldots $$
$$ (0, \ldots, 0, g_0 \alpha^{[(n-k-1)r]}, g_1 \alpha^{[(n-k)r]}, \ldots, g_{n-k} \alpha^{[(n-1)r]}) $$
generate $ \widehat{C} $ as an $ \mathbb{F}_{q^n} $-linear vector space. Since $ g_{n-k} \neq 0 $, it follows that these vectors are linearly independent over $ \mathbb{F}_{q^n} $. Hence the result follows from Theorem \ref{generator} and the fact that $ d_H(C) = d_R(E(C)) \geq d_R(\langle E(C) \rangle_{\mathbb{F}_{q^n}}) = d_R(\widehat{C}) $.
\end{proof}

\begin{example}
Consider the repetition cyclic code $ C \subseteq \mathbb{F}_q^n $ generated by $ (1,1, \ldots, 1) $ and assume $ r=1 $. Then $ E(C) $ is the $ \mathbb{F}_q $-linear code generated by $ (\alpha, \alpha^{[1]}, \ldots, $ $ \alpha^{[n-1]}) $, and hence the $ \mathbb{F}_{q^n} $-linear code generated by $ E(C) $ is $ \widehat{C} $, also generated by the same vector.

It holds that $ \dim_{\mathbb{F}_q}(C) = 1 $, $ d_H(C) = n $ and $ C $ is MDS. On the other hand, $ \dim_{\mathbb{F}_{q^n}}(\widehat{C}) = 1 $, $ d_R(\widehat{C}) = n $ and $ \widehat{C} $ is MRD.
\end{example}

\begin{example}
Assume that $ r=1 $ and $ n $ is even, and consider the cyclic code $ C \subseteq \mathbb{F}_q^n $ generated by $ (1,0,1,0, \ldots, 0) $ and $ (0,1,0,1, \ldots, 1) $. Then $ \widehat{C} $ is the $ \mathbb{F}_{q^n} $-linear code generated by $ (\alpha,0, \alpha^{[2]},0, \ldots, 0) $ and $ (0,\alpha^{[1]},0, \alpha^{[3]}, \ldots, \alpha^{[n-1]}) $.

It holds that $ \dim_{\mathbb{F}_q}(C) = 2 $, $ d_H(C) = n/2 $. On the other hand, $ \dim_{\mathbb{F}_{q^n}}(\widehat{C}) $ $ = 2 $, $ d_R(\widehat{C}) = n/2 $. Hence both have the same parameters and none reach the Singleton bounds for the corresponding metrics. Moreover, the minimal generator of $ C $ is $ g(x) = 1 + x^2 + x^4 + \cdots + x^{n-2} $, whereas the minimal generator of $ \widehat{C} $ is $ E(g(x)) = \alpha x + \alpha^{[2]} x^{[2]} + \alpha^{[4]} x^{[4]} + \cdots + \alpha^{[n-2]}x^{[n-2]} $. 
\end{example}

\section*{Acknowledgement}

The author wishes to thank Olav Geil, Ruud Pellikaan and Diego Ruano for fruitful discussions and careful reading of the manuscript. The author also gratefully acknowledges the support from The Danish Council for Independent Research (Grant No. DFF-4002-00367).


\end{document}